\documentclass[11pt,a4paper]{article}

\usepackage[utf8]{inputenc}
\usepackage[T1]{fontenc}
\usepackage{lmodern}
\usepackage[margin=1in]{geometry}
\usepackage{amsmath,amssymb,amsthm,amsfonts}
\usepackage{mathtools}
\usepackage{algorithm}
\usepackage{algpseudocode}
\usepackage{graphicx}
\graphicspath{{./}{./figures/}}
\usepackage{subcaption}
\usepackage{booktabs}
\usepackage{multirow}
\usepackage{xcolor}
\usepackage{listings}
\usepackage{hyperref}
\usepackage{cleveref}
\usepackage{tikz}
\usetikzlibrary{positioning, matrix, arrows.meta, calc, decorations.pathreplacing}
\usepackage{pgfplots}
\pgfplotsset{compat=1.17}
\usepackage{float}
\usepackage{enumitem}
\usepackage{microtype}
\usepackage{authblk}
\usepackage{tcolorbox}

\theoremstyle{plain}
\newtheorem{theorem}{Theorem}[section]

\newtheorem{proposition}[theorem]{Proposition}
\newtheorem{corollary}[theorem]{Corollary}

\theoremstyle{definition}
\newtheorem{definition}[theorem]{Definition}

\newtheorem{remark}[theorem]{Remark}

\newcommand{\palma}{\textsc{Palma}}
\newcommand{\Rmax}{\mathbb{R}_{\max}}
\newcommand{\Rmin}{\mathbb{R}_{\min}}

\newcommand{\bigO}{\mathcal{O}}
\newcommand{\nnz}{\mathrm{nnz}}

\newcommand{\mat}[1]{\mathbf{#1}}
\newcommand{\vect}[1]{\mathbf{#1}}

\definecolor{codegreen}{rgb}{0,0.6,0}
\definecolor{codegray}{rgb}{0.5,0.5,0.5}
\definecolor{codepurple}{rgb}{0.58,0,0.82}
\definecolor{backcolour}{rgb}{0.97,0.97,0.97}

\lstdefinestyle{cstyle}{
    backgroundcolor=\color{backcolour},
    commentstyle=\color{codegreen},
    keywordstyle=\color{blue},
    numberstyle=\tiny\color{codegray},
    stringstyle=\color{codepurple},
    basicstyle=\ttfamily\small,
    breakatwhitespace=false,
    breaklines=true,
    captionpos=b,
    keepspaces=true,
    numbers=left,
    numbersep=5pt,
    showspaces=false,
    showstringspaces=false,
    showtabs=false,
    tabsize=2,
    language=C
}
\title{%
    \textbf{PALMA: A Lightweight Tropical Algebra Library for ARM-Based Embedded Systems}\\[0.5em]
    \large SIMD-Accelerated Semiring Linear Algebra with Spectral Analysis and Embedded Case Studies
}

\author[1,2,3]{Gnankan Landry Regis N'guessan}

\affil[1]{Axiom Research Group}
\affil[2]{Department of Applied Mathematics and Computational Science,
The Nelson Mandela African Institution of Science and Technology (NM-AIST),
Arusha, Tanzania}
\affil[3]{African Institute for Mathematical Sciences (AIMS),
Research and Innovation Centre (RIC), Kigali, Rwanda}
\affil[ ]{\textit{Email: rnguessan@aimsric.org}}

\date{}

\begin{document}

\maketitle

\begin{abstract}
Tropical algebra, including max-plus, min-plus, and related idempotent semirings, provides a unifying framework in which many optimization problems that are nonlinear in classical algebra become linear. This property makes tropical methods particularly well suited for shortest paths, scheduling, throughput analysis, and discrete event systems. Despite their theoretical maturity and practical relevance, existing tropical algebra implementations primarily target desktop or server environments and remain largely inaccessible on resource-constrained embedded platforms, where such optimization problems are most acute.
We present \palma{} (\textbf{P}arallel \textbf{A}lgebra \textbf{L}ibrary for \textbf{M}ax-plus \textbf{A}pplications), a lightweight, dependency-free C library that brings tropical linear algebra to ARM-based embedded systems. \palma{} implements a generic semiring abstraction with SIMD-accelerated kernels, enabling a single computational framework to support shortest paths, bottleneck paths, reachability, scheduling, and throughput analysis. The library supports five tropical semirings, dense and sparse (CSR) representations, tropical closure, and spectral analysis via maximum cycle mean computation.
We evaluate \palma{} on a Raspberry Pi 4 and demonstrate peak performance of 2,274 MOPS, speedups of up to 11.9 times over classical Bellman-Ford for single-source shortest paths, and sub-10 microsecond scheduling solves for real-time control workloads. Case studies in UAV control, IoT routing, and manufacturing systems show that tropical algebra enables efficient, predictable, and unified optimization directly on embedded hardware. \palma{} is released as open-source software under the MIT license.

\end{abstract}

\section{Introduction}
\label{sec:introduction}

\subsection{Motivation and Background}

The field of tropical mathematics has emerged as a powerful framework for solving optimization problems by transforming them into algebraic operations over idempotent semirings. In tropical algebra, conventional addition is replaced by taking the maximum (or minimum), and conventional multiplication is replaced by addition. This seemingly simple redefinition has profound consequences: many problems that are inherently nonlinear in classical mathematics, such as finding shortest paths, optimizing schedules, and computing system throughput, become \emph{linear over the tropical semiring}, expressible as matrix products and fixed-point iterations \cite{baccelli1992synchronization, butkovic2010max}.

Consider the classical shortest path problem: given a weighted graph, find the path of minimum total weight between two vertices. In standard algebra, this requires iterative comparison and selection operations that defy matrix formulation. However, in the min-plus tropical semiring, the shortest path computation reduces to matrix multiplication:
\begin{equation}
    d_{ij}^{(k)} = \min_{\ell}\left(d_{i\ell}^{(k-1)} + w_{\ell j}\right)
\end{equation}
where $\min$ plays the role of addition and $+$ plays the role of multiplication. This ``linearization'' enables the direct application of powerful linear algebraic techniques to optimization problems \cite{gondran2008graphs, mohri2002semiring}. In \palma{}, this recurrence is expressed as a single API call: \texttt{palma\_matvec(A, d, d\_new, PALMA\_MINPLUS)}, and the all-pairs solution is simply \texttt{palma\_matrix\_closure(A, PALMA\_MINPLUS)}.

Despite four decades of theoretical development since the foundational work of Cuninghame-Green \cite{cuninghame1979minimax} and the comprehensive treatise by Baccelli et al. \cite{baccelli1992synchronization}, practical implementations of tropical algebra remain surprisingly limited. Existing tools such as the ScicosLab MaxPlus toolbox \cite{gaubert2009max} and the polymake system \cite{gawrilow2000polymake, hampe2018tropical} target desktop computing environments and are unsuitable for resource-constrained embedded systems. This gap is particularly problematic because embedded systems (robotics, autonomous vehicles, industrial controllers, IoT networks) are precisely where real-time scheduling and optimal routing are most critical.

\subsection{The Embedded Systems Challenge}

Modern embedded platforms, exemplified by the Raspberry Pi family and similar ARM-based systems, present unique constraints:

\begin{enumerate}[leftmargin=*]
    \item \textbf{Limited computational resources}: While modern embedded processors like the ARM Cortex-A72 are capable, they operate at significantly lower frequencies than desktop CPUs and have smaller caches.
    
    \item \textbf{Memory constraints}: Embedded systems often have limited RAM, making memory-efficient data structures essential.
    
    \item \textbf{Real-time requirements}: Many embedded applications require deterministic execution times, often with deadlines in the microsecond to millisecond range.
    
    \item \textbf{Power efficiency}: Battery-powered and thermally-constrained systems demand operations that minimize energy consumption.
    
    \item \textbf{Deployment simplicity}: Embedded software should minimize external dependencies and be easily portable across platforms.
\end{enumerate}

These constraints motivate the development of \palma{}, a tropical algebra library designed from the ground up for embedded systems. By leveraging ARM NEON SIMD instructions for parallelism, sparse matrix representations for memory efficiency, and pure integer arithmetic for predictability, \palma{} brings the power of tropical mathematics to platforms where it can have the greatest practical impact.

\subsection{Contributions}

This paper makes the following contributions:

\begin{enumerate}[leftmargin=*]
    \item \textbf{Mathematical Foundations}: We provide a comprehensive, self-contained treatment of tropical algebra, including rigorous definitions, key theorems with proofs, and the fundamental algorithms. This serves both as a tutorial and as a reference for the implementation.
    
    \item \textbf{Library Design}: We present \palma{}, a complete tropical algebra library featuring:
    \begin{itemize}
        \item Five tropical semirings: max-plus ($\Rmax$), min-plus ($\Rmin$), max-min, min-max, and Boolean
        \item Dense and sparse (CSR format) matrix representations
        \item ARM NEON SIMD-optimized operations
        \item Tropical eigenvalue computation via Karp's algorithm
        \item High-level APIs for scheduling and graph algorithms
    \end{itemize}
    
    \item \textbf{Implementation}: We describe the implementation in approximately 2,000 lines of portable C99 code, including detailed discussion of NEON optimization strategies and sparse matrix techniques.
    
    \item \textbf{Experimental Evaluation}: We provide comprehensive benchmarks on Raspberry Pi 4, measuring:
    \begin{itemize}
        \item NEON vs. scalar performance across matrix sizes
        \item Scalability analysis from $16 \times 16$ to $1024 \times 1024$ matrices
        \item Comparison of all five semirings
        \item Dense vs. sparse crossover analysis
        \item Performance comparison with classical algorithms (Floyd-Warshall, Bellman-Ford)
    \end{itemize}
    
    \item \textbf{Case Studies}: We demonstrate practical applications through three detailed case studies:
    \begin{itemize}
        \item Real-time drone control system scheduling
        \item IoT sensor network routing optimization
        \item Manufacturing production line throughput analysis
    \end{itemize}
\end{enumerate}

\paragraph{Platform and Terminology Context.}
Throughout this paper, \emph{ARM-based embedded systems} refer to energy-efficient computing platforms built around ARM architectures commonly used in robotics, IoT, and cyber-physical systems, such as system-on-chip devices in the ARM Cortex-A series deployed in autonomous platforms, industrial controllers, and edge computing nodes. These processors emphasize predictable execution, moderate parallelism, and constrained memory resources rather than high clock frequencies or large caches. \emph{SIMD} (Single Instruction, Multiple Data) denotes a parallel execution model in which one instruction operates simultaneously on multiple data elements. On ARM platforms, SIMD functionality is provided by the \emph{NEON} extension, which offers 128-bit vector registers supporting parallel operations on four 32-bit integers. In \palma{}, NEON is used exclusively as a low-level optimization mechanism to accelerate core tropical semiring operations, such as max, min, and addition, while all algebraic definitions and algorithmic semantics remain independent of the underlying hardware and numeric representation.

\paragraph{Extensibility Beyond ARM.}
While the current implementation of \palma{} targets ARM-based embedded platforms, the library is designed around a hardware-agnostic semiring abstraction. This architectural choice enables future support for additional embedded architectures, such as RISC-V, DSP-based systems, and microcontroller-class platforms with SIMD or vector extensions, without altering the mathematical or algorithmic core. The long-term goal of \palma{} is to evolve as a portable computational foundation for real-world embedded optimization problems across heterogeneous hardware environments.

\subsection{Paper Organization}

The remainder of this paper is organized as follows. Section~\ref{sec:background} establishes the mathematical foundations of tropical algebra, including semiring theory, tropical linear algebra, and eigenvalue computation. Section~\ref{sec:design} describes the \palma{} library architecture and design decisions. Section~\ref{sec:implementation} details the implementation, with emphasis on NEON optimization and sparse matrix techniques. Section~\ref{sec:algorithms} presents the core algorithms implemented in \palma{}. Section~\ref{sec:experiments} provides comprehensive experimental evaluation on Raspberry Pi 4. Section~\ref{sec:casestudies} presents three detailed case studies demonstrating practical applications. Section~\ref{sec:related} discusses related work. Finally, Section~\ref{sec:conclusion} concludes the paper and outlines future directions.

\section{Mathematical Foundations}
\label{sec:background}

This section provides a rigorous treatment of tropical algebra, establishing the theoretical foundation for \palma{}. We begin with the general theory of semirings, then specialize to tropical semirings, develop tropical linear algebra, and conclude with spectral theory and the crucial connection to graph algorithms.

\paragraph{Notation Convention.}
Throughout this paper, we use $(S, \oplus, \otimes, \mathbf{0}, \mathbf{1})$ to denote a generic idempotent semiring, with $\oplus$ as ``addition'' and $\otimes$ as ``multiplication.'' We then specialize to five concrete semirings: \emph{max-plus} $(\Rmax)$, \emph{min-plus} $(\Rmin)$, \emph{max-min}, \emph{min-max}, and \emph{Boolean}. Matrices are denoted by bold uppercase letters ($\mat{A}, \mat{B}$), vectors by bold lowercase ($\vect{x}, \vect{y}$), and scalars by italic ($a, \lambda$). The tropical matrix product uses the same $\otimes$ symbol as scalar multiplication when context is clear.

\subsection{Semiring Theory}
\label{subsec:semirings}

\begin{definition}[Semiring]
\label{def:semiring}
A \emph{semiring} is an algebraic structure $(S, \oplus, \otimes, \mathbf{0}, \mathbf{1})$ where $S$ is a set, $\oplus$ and $\otimes$ are binary operations on $S$, and $\mathbf{0}, \mathbf{1} \in S$, satisfying:
\begin{enumerate}[label=(\roman*)]
    \item $(S, \oplus, \mathbf{0})$ is a commutative monoid:
    \begin{itemize}
        \item Associativity: $(a \oplus b) \oplus c = a \oplus (b \oplus c)$
        \item Commutativity: $a \oplus b = b \oplus a$
        \item Identity: $a \oplus \mathbf{0} = a$
    \end{itemize}
    \item $(S, \otimes, \mathbf{1})$ is a monoid:
    \begin{itemize}
        \item Associativity: $(a \otimes b) \otimes c = a \otimes (b \otimes c)$
        \item Identity: $a \otimes \mathbf{1} = \mathbf{1} \otimes a = a$
    \end{itemize}
    \item Multiplication distributes over addition:
    \begin{itemize}
        \item Left: $a \otimes (b \oplus c) = (a \otimes b) \oplus (a \otimes c)$
        \item Right: $(a \oplus b) \otimes c = (a \otimes c) \oplus (b \otimes c)$
    \end{itemize}
    \item The zero element annihilates:
    \begin{itemize}
        \item $a \otimes \mathbf{0} = \mathbf{0} \otimes a = \mathbf{0}$
    \end{itemize}
\end{enumerate}
\end{definition}

\begin{definition}[Idempotent Semiring]
A semiring is \emph{idempotent} if $a \oplus a = a$ for all $a \in S$.
\end{definition}

Idempotent semirings are central to tropical mathematics because they induce a natural partial order.

\begin{proposition}[Natural Order]
\label{prop:natural-order}
In an idempotent semiring, the relation $a \leq b \iff a \oplus b = b$ defines a partial order on $S$.
\end{proposition}

\begin{proof}
We verify the order axioms:
\begin{itemize}
    \item \textbf{Reflexivity}: $a \oplus a = a$ by idempotency, so $a \leq a$.
    \item \textbf{Antisymmetry}: If $a \leq b$ and $b \leq a$, then $a \oplus b = b$ and $b \oplus a = a$. By commutativity, $a = b \oplus a = a \oplus b = b$.
    \item \textbf{Transitivity}: If $a \leq b$ and $b \leq c$, then $a \oplus b = b$ and $b \oplus c = c$. We have $a \oplus c = a \oplus (b \oplus c) = (a \oplus b) \oplus c = b \oplus c = c$, so $a \leq c$. \qedhere
\end{itemize}
\end{proof}

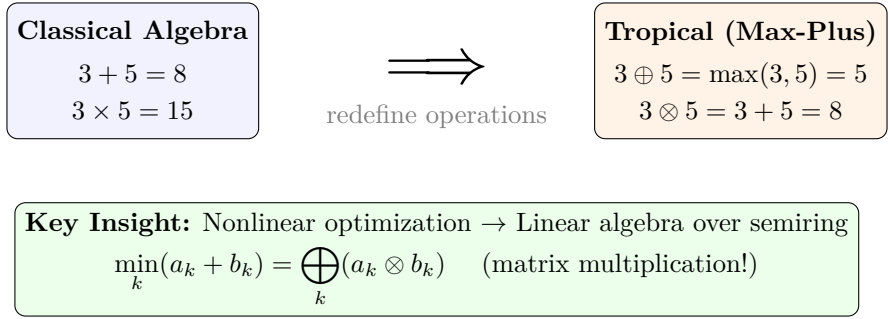
\begin{figure}[htbp]
\centering
\begin{tikzpicture}[
    box/.style={draw, rounded corners, minimum width=3.2cm, minimum height=1.8cm, align=center, font=\small},
    op/.style={font=\Large\bfseries, text=blue!70!black},
    result/.style={font=\large\bfseries, text=red!70!black},
    label/.style={font=\small\bfseries, text=gray}
]
\node[box, fill=blue!5] (classical) at (-4, 0) {
    \textbf{Classical Algebra}\\[4pt]
    $3 + 5 = 8$\\[2pt]
    $3 \times 5 = 15$
};

\node[font=\Huge] at (0, 0) {$\Longrightarrow$};
\node[font=\small, text=gray] at (0, -0.6) {redefine operations};

\node[box, fill=orange!10] (tropical) at (4, 0) {
    \textbf{Tropical (Max-Plus)}\\[4pt]
    $3 \oplus 5 = \max(3,5) = 5$\\[2pt]
    $3 \otimes 5 = 3 + 5 = 8$
};

\node[box, fill=green!8, minimum width=10cm, minimum height=1.4cm] at (0, -2.5) {
    \textbf{Key Insight:} Nonlinear optimization $\to$ Linear algebra over semiring\\[2pt]
    $\displaystyle\min_{k}(a_k + b_k) = \bigoplus_k (a_k \otimes b_k)$ \quad (matrix multiplication!)
};

\end{tikzpicture}
\caption{Classical algebra vs.\ tropical (max-plus) algebra. The redefinition of addition as $\max$ and multiplication as $+$ transforms optimization problems into linear algebraic operations.}
\label{fig:classical-vs-tropical}
\end{figure}

\subsection{Tropical Semirings}
\label{subsec:tropical-semirings}

We now define the five tropical semirings implemented in \palma{}.

\begin{definition}[Max-Plus Semiring $\Rmax$]
\label{def:max-plus}
The \emph{max-plus semiring} is $\Rmax = (\mathbb{R} \cup \{-\infty\}, \max, +, -\infty, 0)$ where:
\begin{itemize}
    \item $a \oplus b = \max(a, b)$
    \item $a \otimes b = a + b$
    \item $\mathbf{0} = \varepsilon = -\infty$ (absorbing element)
    \item $\mathbf{1} = e = 0$ (identity element)
\end{itemize}
\end{definition}

\begin{definition}[Min-Plus Semiring $\Rmin$]
\label{def:min-plus}
The \emph{min-plus semiring} is $\Rmin = (\mathbb{R} \cup \{+\infty\}, \min, +, +\infty, 0)$ where:
\begin{itemize}
    \item $a \oplus b = \min(a, b)$
    \item $a \otimes b = a + b$
    \item $\mathbf{0} = +\infty$
    \item $\mathbf{1} = 0$
\end{itemize}
\end{definition}

\begin{definition}[Max-Min Semiring]
\label{def:max-min}
The \emph{max-min semiring} (also called the \emph{bottleneck semiring}) is $(\mathbb{R} \cup \{-\infty, +\infty\}, \max, \min, -\infty, +\infty)$ where:
\begin{itemize}
    \item $a \oplus b = \max(a, b)$
    \item $a \otimes b = \min(a, b)$
    \item $\mathbf{0} = -\infty$
    \item $\mathbf{1} = +\infty$
\end{itemize}
\end{definition}

\begin{definition}[Min-Max Semiring]
\label{def:min-max}
The \emph{min-max semiring} is $(\mathbb{R} \cup \{-\infty, +\infty\}, \min, \max, +\infty, -\infty)$ where:
\begin{itemize}
    \item $a \oplus b = \min(a, b)$
    \item $a \otimes b = \max(a, b)$
    \item $\mathbf{0} = +\infty$
    \item $\mathbf{1} = -\infty$
\end{itemize}
\end{definition}

\begin{definition}[Boolean Semiring]
\label{def:boolean}
The \emph{Boolean semiring} is $(\{0, 1\}, \vee, \wedge, 0, 1)$ where:
\begin{itemize}
    \item $a \oplus b = a \vee b$ (logical OR)
    \item $a \otimes b = a \wedge b$ (logical AND)
    \item $\mathbf{0} = 0$ (false)
    \item $\mathbf{1} = 1$ (true)
\end{itemize}
\end{definition}

\begin{theorem}[Tropical Semiring Properties]
\label{thm:tropical-properties}
All five semirings defined above are valid semirings. Moreover, all are idempotent.
\end{theorem}

\begin{proof}
We prove the result for $\Rmax$; the other cases are analogous.

\textbf{Additive monoid}: $(\mathbb{R} \cup \{-\infty\}, \max, -\infty)$ is clearly a commutative monoid since $\max$ is associative and commutative, and $\max(a, -\infty) = a$.

\textbf{Multiplicative monoid}: $(\mathbb{R} \cup \{-\infty\}, +, 0)$ is a monoid with the convention $a + (-\infty) = -\infty$ for all $a$.

\textbf{Distributivity}: For all $a, b, c \in \Rmax$:
\begin{align*}
    a \otimes (b \oplus c) &= a + \max(b, c) = \max(a + b, a + c) = (a \otimes b) \oplus (a \otimes c)
\end{align*}
The equality $a + \max(b, c) = \max(a + b, a + c)$ holds because addition preserves order.

\textbf{Annihilation}: $a \otimes (-\infty) = a + (-\infty) = -\infty = \mathbf{0}$.

\textbf{Idempotency}: $\max(a, a) = a$, so $a \oplus a = a$.
\end{proof}

\begin{table}[htbp]
\centering
\caption{Summary of tropical semirings implemented in \palma{}}
\label{tab:semirings}
\begin{tabular}{lccccl}
\toprule
\textbf{Semiring} & $\oplus$ & $\otimes$ & $\mathbf{0}$ & $\mathbf{1}$ & \textbf{Primary Application} \\
\midrule
Max-Plus & $\max$ & $+$ & $-\infty$ & $0$ & Scheduling, longest paths \\
Min-Plus & $\min$ & $+$ & $+\infty$ & $0$ & Shortest paths, Dijkstra \\
Max-Min & $\max$ & $\min$ & $-\infty$ & $+\infty$ & Bottleneck/bandwidth paths \\
Min-Max & $\min$ & $\max$ & $+\infty$ & $-\infty$ & Reliability paths \\
Boolean & $\vee$ & $\wedge$ & $0$ & $1$ & Reachability, connectivity \\
\bottomrule
\end{tabular}
\end{table}

\subsection{Tropical Linear Algebra}
\label{subsec:tropical-linear-algebra}

Having established the semiring structure, we now develop tropical linear algebra over an arbitrary idempotent semiring $(S, \oplus, \otimes, \mathbf{0}, \mathbf{1})$.

\begin{definition}[Tropical Matrix Operations]
Let $\mat{A} \in S^{m \times n}$ and $\mat{B} \in S^{n \times p}$ be matrices over $S$. The \emph{tropical matrix product} $\mat{C} = \mat{A} \otimes \mat{B} \in S^{m \times p}$ is defined by:
\begin{equation}
    C_{ij} = \bigoplus_{k=1}^{n} (A_{ik} \otimes B_{kj}) = \bigoplus_{k=1}^{n} A_{ik} \otimes B_{kj}
\end{equation}
\end{definition}

For the max-plus semiring, this becomes:
\begin{equation}
    C_{ij} = \max_{k=1}^{n} (A_{ik} + B_{kj})
\end{equation}

For the min-plus semiring:
\begin{equation}
    C_{ij} = \min_{k=1}^{n} (A_{ik} + B_{kj})
\end{equation}

\begin{figure}[htbp]
\centering
\begin{tikzpicture}[
    mat/.style={matrix of math nodes, nodes in empty cells, left delimiter={[}, right delimiter={]}, 
                nodes={minimum width=1.1cm, minimum height=0.7cm, font=\small}},
    highlight/.style={fill=yellow!40},
    result/.style={fill=green!30},
    label/.style={font=\small\bfseries}
]

\matrix[mat] (A) at (0,0) {
    2 & 3 & |[highlight]| 1 \\
    5 & 0 & 4 \\
};
\node[label, above=0.3cm of A] {$\mat{A}$};

\node at (2.2, 0) {\Large$\otimes$};

\matrix[mat] (B) at (4.2, 0) {
    |[highlight]| 1 & 2 \\
    |[highlight]| 4 & 0 \\
    |[highlight]| 2 & 3 \\
};
\node[label, above=0.3cm of B] {$\mat{B}$};

\node at (6.2, 0) {\Large$=$};

\matrix[mat] (C) at (8.2, 0) {
    |[result]| ? & \cdot \\
    \cdot & \cdot \\
};
\node[label, above=0.3cm of C] {$\mat{C}$};

\node[draw, rounded corners, fill=blue!5, text width=12cm, align=left, font=\small] at (4, -2.2) {
    \textbf{Max-Plus:} $C_{11} = \max(A_{11}+B_{11},\, A_{12}+B_{21},\, A_{13}+B_{31}) = \max(2+1,\, 3+4,\, 1+2) = \max(3, 7, 3) = \mathbf{7}$\\[4pt]
    \textbf{Min-Plus:} $C_{11} = \min(A_{11}+B_{11},\, A_{12}+B_{21},\, A_{13}+B_{31}) = \min(2+1,\, 3+4,\, 1+2) = \min(3, 7, 3) = \mathbf{3}$
};

\draw[->, thick, red!70!black, dashed] (A-1-1.east) to[out=0, in=180] node[above, font=\tiny] {$2+1$} (B-1-1.west);
\draw[->, thick, blue!70!black, dashed] (A-1-2.east) to[out=0, in=180] node[above, font=\tiny, yshift=2pt] {$3+4$} (B-2-1.west);
\draw[->, thick, green!60!black, dashed] (A-1-3.east) to[out=0, in=180] node[below, font=\tiny] {$1+2$} (B-3-1.west);

\end{tikzpicture}
\caption{Tropical matrix multiplication illustrated. To compute $C_{ij}$, we take the tropical sum ($\max$ or $\min$) over all paths from row $i$ of $\mat{A}$ to column $j$ of $\mat{B}$, where each path contributes the tropical product ($+$) of its elements.}
\label{fig:tropical-matmul}
\end{figure}
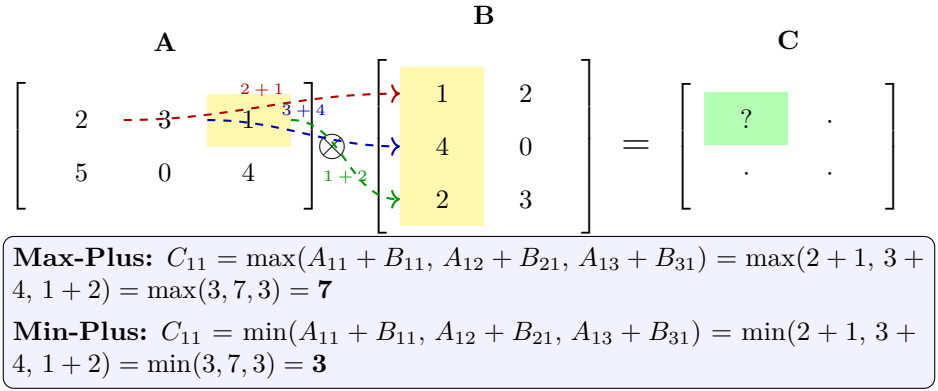

\begin{definition}[Tropical Matrix-Vector Product]
For $\mat{A} \in S^{m \times n}$ and $\vect{x} \in S^n$, the product $\vect{y} = \mat{A} \otimes \vect{x} \in S^m$ is:
\begin{equation}
    y_i = \bigoplus_{j=1}^{n} A_{ij} \otimes x_j
\end{equation}
\end{definition}

\begin{definition}[Identity Matrix]
The \emph{tropical identity matrix} $\mat{E} \in S^{n \times n}$ is:
\begin{equation}
    E_{ij} = \begin{cases}
        \mathbf{1} & \text{if } i = j \\
        \mathbf{0} & \text{if } i \neq j
    \end{cases}
\end{equation}
\end{definition}

\begin{proposition}[Matrix Multiplication Properties]
\label{prop:matrix-mult}
Tropical matrix multiplication satisfies:
\begin{enumerate}[label=(\roman*)]
    \item \textbf{Associativity}: $(\mat{A} \otimes \mat{B}) \otimes \mat{C} = \mat{A} \otimes (\mat{B} \otimes \mat{C})$
    \item \textbf{Identity}: $\mat{A} \otimes \mat{E} = \mat{E} \otimes \mat{A} = \mat{A}$
    \item \textbf{Distributivity}: $\mat{A} \otimes (\mat{B} \oplus \mat{C}) = (\mat{A} \otimes \mat{B}) \oplus (\mat{A} \otimes \mat{C})$
\end{enumerate}
However, in general:
\begin{enumerate}[label=(\roman*)]
    \setcounter{enumi}{3}
    \item Matrix multiplication is \textbf{not commutative}: $\mat{A} \otimes \mat{B} \neq \mat{B} \otimes \mat{A}$
\end{enumerate}
\end{proposition}

\begin{proof}
These properties follow directly from the semiring axioms applied element-wise. For associativity:
\begin{align*}
    ((\mat{A} \otimes \mat{B}) \otimes \mat{C})_{ij} &= \bigoplus_k \left(\bigoplus_\ell A_{i\ell} \otimes B_{\ell k}\right) \otimes C_{kj} \\
    &= \bigoplus_k \bigoplus_\ell (A_{i\ell} \otimes B_{\ell k}) \otimes C_{kj} \\
    &= \bigoplus_\ell \bigoplus_k A_{i\ell} \otimes (B_{\ell k} \otimes C_{kj}) \\
    &= \bigoplus_\ell A_{i\ell} \otimes \left(\bigoplus_k B_{\ell k} \otimes C_{kj}\right) \\
    &= (\mat{A} \otimes (\mat{B} \otimes \mat{C}))_{ij}
\end{align*}
where we used distributivity and the commutativity of $\oplus$.
\end{proof}

\begin{definition}[Matrix Powers]
For a square matrix $\mat{A} \in S^{n \times n}$, we define:
\begin{equation}
    \mat{A}^0 = \mat{E}, \quad \mat{A}^k = \mat{A}^{k-1} \otimes \mat{A} \text{ for } k \geq 1
\end{equation}
\end{definition}

\begin{theorem}[Powers and Paths]
\label{thm:powers-paths}
Let $\mat{A}$ be the adjacency matrix of a weighted directed graph $G$, where $A_{ij}$ represents the weight of the edge from $i$ to $j$ (or $\mathbf{0}$ if no edge exists). Then:
\begin{enumerate}[label=(\roman*)]
    \item In the min-plus semiring: $(\mat{A}^k)_{ij}$ equals the weight of the shortest path from $i$ to $j$ using exactly $k$ edges.
    \item In the max-plus semiring: $(\mat{A}^k)_{ij}$ equals the weight of the longest path from $i$ to $j$ using exactly $k$ edges.
    \item In the max-min semiring: $(\mat{A}^k)_{ij}$ equals the maximum bottleneck (minimum edge weight on path) from $i$ to $j$ using exactly $k$ edges.
    \item In the Boolean semiring: $(\mat{A}^k)_{ij} = 1$ iff there exists a path from $i$ to $j$ using exactly $k$ edges.
\end{enumerate}
\end{theorem}

\begin{proof}
We prove (i) by induction; the other cases are analogous.

\textbf{Base case} ($k=1$): $(\mat{A}^1)_{ij} = A_{ij}$ is the weight of the direct edge from $i$ to $j$, which is the only 1-edge path.

\textbf{Inductive step}: Assume the result holds for $k-1$. Then:
\begin{align*}
    (\mat{A}^k)_{ij} &= \bigoplus_{\ell=1}^n (\mat{A}^{k-1})_{i\ell} \otimes A_{\ell j} \\
    &= \min_{\ell=1}^n \left((\mat{A}^{k-1})_{i\ell} + A_{\ell j}\right)
\end{align*}
By the inductive hypothesis, $(\mat{A}^{k-1})_{i\ell}$ is the shortest $(k-1)$-edge path from $i$ to $\ell$. Adding the edge $(\ell, j)$ gives a $k$-edge path from $i$ to $j$ through $\ell$. Taking the minimum over all possible intermediate vertices $\ell$ yields the shortest $k$-edge path.
\end{proof}

\begin{figure}[htbp]
\centering
\begin{tikzpicture}[
    node distance=2cm,
    vertex/.style={circle, draw, thick, minimum size=0.8cm, font=\small\bfseries},
    edge/.style={->, thick, >=stealth},
    gedge/.style={edge, gray!50},
    path1/.style={edge, red!70!black, line width=1.5pt},
    path2/.style={edge, blue!70!black, line width=1.5pt},
    path3/.style={edge, green!60!black, line width=1.5pt},
    mat/.style={matrix of math nodes, nodes in empty cells, left delimiter={[}, right delimiter={]}, 
                nodes={minimum width=0.6cm, minimum height=0.5cm, font=\scriptsize}},
    label/.style={font=\footnotesize}
]

\node[vertex, fill=red!20] (1) at (0, 0) {1};
\node[vertex, fill=yellow!20] (2) at (2, 1) {2};
\node[vertex, fill=yellow!20] (3) at (2, -1) {3};
\node[vertex, fill=green!20] (4) at (4, 0) {4};

\draw[gedge] (1) -- node[above, font=\scriptsize] {3} (2);
\draw[gedge] (1) -- node[below, font=\scriptsize] {5} (3);
\draw[gedge] (2) -- node[above, font=\scriptsize] {2} (4);
\draw[gedge] (3) -- node[below, font=\scriptsize] {1} (4);
\draw[gedge] (2) -- node[right, font=\scriptsize] {4} (3);

\node[font=\bfseries] at (2, 2) {Weighted Graph $G$};

\node[font=\small\bfseries] at (7, 1.8) {$\mat{A}$ (1-edge paths)};
\matrix[mat] (matA) at (7, 0.5) {
    \infty & 3 & 5 & \infty \\
    \infty & \infty & 4 & 2 \\
    \infty & \infty & \infty & 1 \\
    \infty & \infty & \infty & \infty \\
};

\node[font=\small\bfseries] at (11, 1.8) {$\mat{A}^2$ (2-edge paths)};
\matrix[mat] (matA2) at (11, 0.5) {
    \infty & \infty & 7 & 5 \\
    \infty & \infty & \infty & 5 \\
    \infty & \infty & \infty & \infty \\
    \infty & \infty & \infty & \infty \\
};

\node[draw, rounded corners, fill=blue!5, text width=13cm, align=left, font=\scriptsize] at (6.5, -2) {
    \textbf{Min-Plus interpretation} (shortest paths):\\[2pt]
    $(\mat{A}^1)_{14} = \infty$ \quad (no direct 1-edge path from 1 to 4)\\[1pt]
    $(\mat{A}^2)_{14} = \min(3+2,\, 5+1) = \min(5, 6) = 5$ \quad (shortest 2-edge path: $1 \to 2 \to 4$ with weight 5)\\[1pt]
    $(\mat{A}^3)_{14} = \min(3+4+1) = 8$ \quad (3-edge path: $1 \to 2 \to 3 \to 4$)
};

\draw[path1, bend left=10] (1) to node[above, font=\tiny, text=red!70!black] {$3+2=5$} (2);
\draw[path1] (2) to (4);

\end{tikzpicture}
\caption{Matrix powers correspond to paths of specific lengths. In the min-plus semiring, $(\mat{A}^k)_{ij}$ gives the weight of the shortest path from $i$ to $j$ using exactly $k$ edges. Here, $(\mat{A}^2)_{14} = 5$ corresponds to the path $1 \to 2 \to 4$ with total weight $3+2=5$.}
\label{fig:powers-paths}
\end{figure}
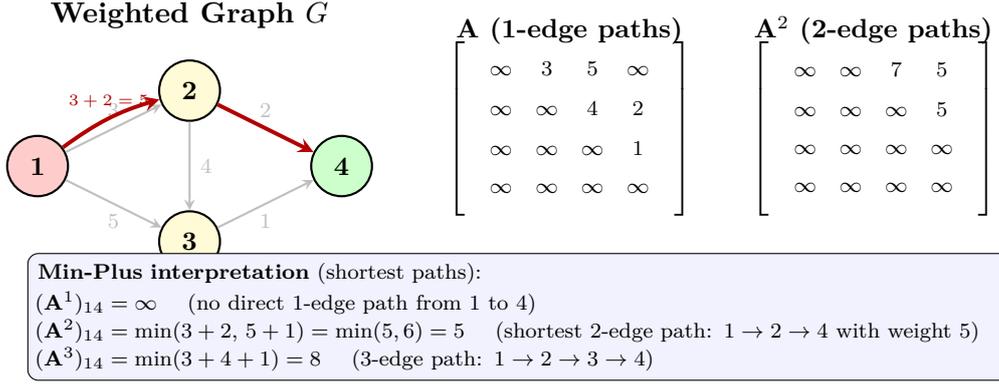

\subsection{The Kleene Star (Tropical Closure)}
\label{subsec:kleene-star}

\begin{definition}[Kleene Star / Tropical Closure]
\label{def:kleene-star}
For a square matrix $\mat{A} \in S^{n \times n}$, the \emph{Kleene star} or \emph{tropical closure} is:
\begin{equation}
    \mat{A}^* = \bigoplus_{k=0}^{\infty} \mat{A}^k = \mat{E} \oplus \mat{A} \oplus \mat{A}^2 \oplus \mat{A}^3 \oplus \cdots
\end{equation}
\end{definition}

\begin{theorem}[Convergence of Kleene Star]
\label{thm:kleene-convergence}
For the min-plus semiring with \emph{no negative-weight cycles}, or for any semiring where $\mat{A}$ represents an acyclic graph:
\begin{equation}
    \mat{A}^* = \bigoplus_{k=0}^{n-1} \mat{A}^k
\end{equation}
That is, the infinite sum converges after $n-1$ terms. The condition ``non-negative edge weights'' is a sufficient (but not necessary) condition for ``no negative cycles.''
\end{theorem}

\begin{proof}
In a graph with $n$ vertices, any simple path has at most $n-1$ edges. For the min-plus semiring without negative cycles, traversing a cycle cannot decrease path weight, so optimal paths are simple. Therefore, $(\mat{A}^k)_{ij} = (\mat{A}^{n-1})_{ij}$ for all $k \geq n-1$, and by idempotency of $\min$, additional terms do not change the sum.
\end{proof}

\begin{theorem}[Kleene Star and All-Pairs Paths]
\label{thm:kleene-apsp}
For the min-plus semiring, $(\mat{A}^*)_{ij}$ equals the weight of the shortest path from $i$ to $j$ (using any number of edges). This is the solution to the all-pairs shortest paths (APSP) problem.
\end{theorem}

\begin{corollary}[Floyd-Warshall as Tropical Closure]
The Floyd-Warshall algorithm computes the Kleene star of the adjacency matrix in the min-plus semiring.
\end{corollary}

\begin{figure}[htbp]
\centering
\begin{tikzpicture}[
    node distance=1.5cm,
    vertex/.style={circle, draw, thick, minimum size=0.7cm, font=\small\bfseries},
    edge/.style={->, thick, >=stealth},
    mat/.style={matrix of math nodes, nodes in empty cells, left delimiter={[}, right delimiter={]}, 
                nodes={minimum width=0.55cm, minimum height=0.45cm, font=\tiny}},
    converge/.style={draw, rounded corners, fill=green!10, minimum width=1.8cm, minimum height=1.2cm},
    label/.style={font=\footnotesize\bfseries}
]

\node[vertex] (a) at (0, 2.5) {1};
\node[vertex] (b) at (1.5, 2.5) {2};
\node[vertex] (c) at (3, 2.5) {3};
\draw[edge] (a) -- node[above, font=\tiny] {2} (b);
\draw[edge] (b) -- node[above, font=\tiny] {3} (c);
\draw[edge, bend left=30] (a) to node[above, font=\tiny] {7} (c);

\node[label] at (-0.5, 0.8) {$\mat{A}^0 = \mat{E}$};
\matrix[mat] (A0) at (0.8, 0.8) {
    0 & \infty & \infty \\
    \infty & 0 & \infty \\
    \infty & \infty & 0 \\
};

\node at (2.1, 0.8) {\large$\oplus$};

\node[label] at (2.5, 0.8) {$\mat{A}^1$};
\matrix[mat] (A1) at (3.7, 0.8) {
    \infty & 2 & 7 \\
    \infty & \infty & 3 \\
    \infty & \infty & \infty \\
};

\node at (5.0, 0.8) {\large$\oplus$};

\node[label] at (5.4, 0.8) {$\mat{A}^2$};
\matrix[mat] (A2) at (6.6, 0.8) {
    \infty & \infty & 5 \\
    \infty & \infty & \infty \\
    \infty & \infty & \infty \\
};

\node at (7.9, 0.8) {\large$\oplus$};

\node at (8.5, 0.8) {$\cdots$};

\node at (9.3, 0.8) {\large$=$};

\node[label] at (9.9, 0.8) {$\mat{A}^*$};
\matrix[mat, nodes={fill=green!15}] (Astar) at (11.1, 0.8) {
    0 & 2 & 5 \\
    \infty & 0 & 3 \\
    \infty & \infty & 0 \\
};

\node[draw, rounded corners, fill=blue!5, text width=13cm, align=left, font=\scriptsize] at (5.5, -1) {
    \textbf{Convergence:} For $n=3$ vertices, $\mat{A}^* = \mat{A}^0 \oplus \mat{A}^1 \oplus \mat{A}^2$ (only $n-1=2$ matrix powers needed).\\[2pt]
    \textbf{Result:} $(\mat{A}^*)_{13} = \min(\infty, 7, 5) = 5$ is the shortest path from 1 to 3 (via vertex 2: $2+3=5$).\\[2pt]
    \textbf{Idempotency:} Once $(\mat{A}^k)_{ij}$ stabilizes, additional terms don't change the $\min$ (since $\min(x,x)=x$).
};

\end{tikzpicture}
\caption{The Kleene star $\mat{A}^* = \bigoplus_{k=0}^{\infty} \mat{A}^k$ converges in $n-1$ iterations. Each power $\mat{A}^k$ captures $k$-edge paths; the tropical sum ($\min$) selects the shortest. Entry $(\mat{A}^*)_{ij}$ gives the all-pairs shortest path from $i$ to $j$.}
\label{fig:kleene-star}
\end{figure}
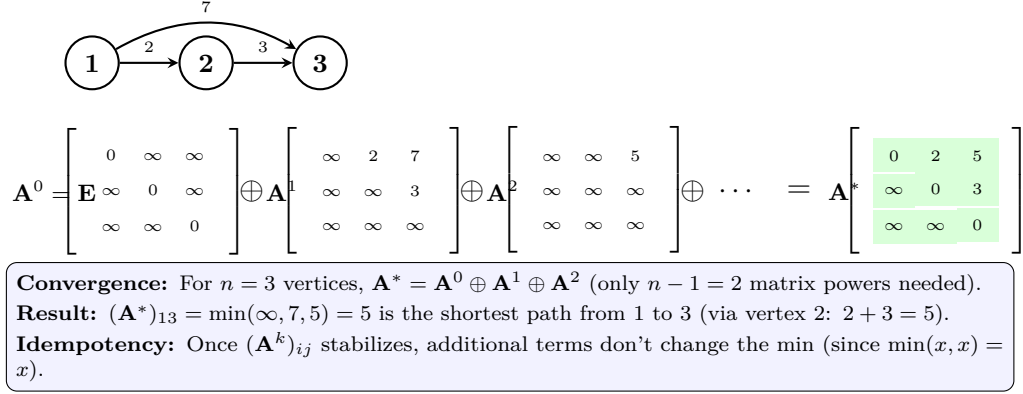

\begin{algorithm}[htbp]
\caption{Tropical Closure (Floyd-Warshall Style)}
\label{alg:closure}
\begin{algorithmic}[1]
\Require Adjacency matrix $\mat{A} \in S^{n \times n}$
\Ensure Closure matrix $\mat{A}^*$
\State $\mat{D} \leftarrow \mat{A}$
\For{$i = 1$ to $n$}
    \State $D_{ii} \leftarrow D_{ii} \oplus \mathbf{1}$ \Comment{Add identity}
\EndFor
\For{$k = 1$ to $n$}
    \For{$i = 1$ to $n$}
        \For{$j = 1$ to $n$}
            \State $D_{ij} \leftarrow D_{ij} \oplus (D_{ik} \otimes D_{kj})$
        \EndFor
    \EndFor
\EndFor
\State \Return $\mat{D}$
\end{algorithmic}
\end{algorithm}

\begin{theorem}[Closure Complexity]
Algorithm~\ref{alg:closure} computes $\mat{A}^*$ in $\bigO(n^3)$ time and $\bigO(n^2)$ space.
\end{theorem}

\subsection{Tropical Eigenvalue Theory}
\label{subsec:eigenvalues}

Tropical eigenvalue theory is fundamental for analyzing periodic systems and computing throughput in manufacturing and scheduling applications \cite{baccelli1992synchronization, heidergott2006max}.

\begin{definition}[Tropical Eigenvalue and Eigenvector]
\label{def:eigenvalue}
A scalar $\lambda \in S$ is a \emph{tropical eigenvalue} of matrix $\mat{A} \in S^{n \times n}$ if there exists a nonzero vector $\vect{v} \in S^n$ (the \emph{eigenvector}) such that:
\begin{equation}
    \mat{A} \otimes \vect{v} = \lambda \otimes \vect{v}
\end{equation}
In the max-plus semiring, this means:
\begin{equation}
    \max_{j=1}^{n}(A_{ij} + v_j) = \lambda + v_i \quad \text{for all } i
\end{equation}
\end{definition}

\begin{definition}[Weighted Digraph]
\label{def:weighted-digraph}
A matrix $\mat{A} \in \Rmax^{n \times n}$ defines a weighted directed graph $G(\mat{A}) = (V, E, w)$ where:
\begin{itemize}
    \item $V = \{1, 2, \ldots, n\}$ is the vertex set
    \item $(i, j) \in E$ iff $A_{ij} \neq -\infty$
    \item $w(i, j) = A_{ij}$ is the edge weight
\end{itemize}
That is, the entry $A_{ij}$ represents the weight of the edge from vertex $i$ to vertex $j$.
\end{definition}

\begin{definition}[Cycle Mean]
For a cycle $\sigma = (i_1, i_2, \ldots, i_k, i_1)$ in $G(\mat{A})$, the \emph{cycle mean} is:
\begin{equation}
    \mu(\sigma) = \frac{1}{k}\sum_{j=1}^{k} A_{i_{j+1}, i_j} = \frac{\text{weight}(\sigma)}{\text{length}(\sigma)}
\end{equation}
where indices are taken modulo $k$.
\end{definition}

\begin{theorem}[Maximum Cycle Mean Theorem]
\label{thm:mcm}
For a matrix $\mat{A} \in \Rmax^{n \times n}$ whose associated graph $G(\mat{A})$ is strongly connected, the unique tropical eigenvalue is:
\begin{equation}
    \lambda(\mat{A}) = \max_{\sigma \in \mathcal{C}} \mu(\sigma)
\end{equation}
where $\mathcal{C}$ is the set of all elementary cycles in $G(\mat{A})$.
\end{theorem}

\begin{proof}[Proof Sketch]
The full proof can be found in \cite{baccelli1992synchronization, butkovic2010max}. The key insight is that for any eigenvector $\vect{v}$:
\begin{equation}
    \lambda = \max_j (A_{ij} + v_j) - v_i
\end{equation}
Summing around any cycle and using the telescoping property of $v_i$ terms shows that $\lambda$ must equal the mean of that cycle. The maximum over all cycles gives the eigenvalue.
\end{proof}

\begin{figure}[htbp]
\centering
\begin{tikzpicture}[
    vertex/.style={circle, draw, thick, minimum size=0.9cm, font=\small\bfseries},
    edge/.style={->, thick, >=stealth},
    cycle1/.style={edge, red!70!black, line width=2pt},
    cycle2/.style={edge, blue!70!black, line width=2pt},
    cycle3/.style={edge, orange!80!black, line width=2pt},
    label/.style={font=\small}
]

\node[vertex, fill=gray!20] (1) at (0, 0) {1};
\node[vertex, fill=gray!20] (2) at (2.5, 1.2) {2};
\node[vertex, fill=gray!20] (3) at (2.5, -1.2) {3};
\node[vertex, fill=gray!20] (4) at (5, 0) {4};

\draw[edge, gray!50] (1) -- node[above left, font=\scriptsize] {4} (2);
\draw[edge, gray!50] (2) -- node[above right, font=\scriptsize] {3} (4);
\draw[edge, gray!50] (1) -- node[below left, font=\scriptsize] {2} (3);
\draw[edge, gray!50] (3) -- node[below right, font=\scriptsize] {5} (4);

\draw[cycle1, bend left=15] (1) to node[above, font=\scriptsize, text=red!70!black] {4} (2);
\draw[cycle1] (2) to node[right, font=\scriptsize, text=red!70!black] {1} (3);
\draw[cycle1, bend left=15] (3) to node[below, font=\scriptsize, text=red!70!black] {3} (1);

\draw[cycle2, bend left=20] (2) to node[above, font=\scriptsize, text=blue!70!black] {3} (4);
\draw[cycle2, bend left=20] (4) to node[below, font=\scriptsize, text=blue!70!black] {2} (2);

\node[draw, rounded corners, fill=white, text width=5.5cm, align=left, font=\scriptsize] at (9, 1) {
    \textbf{Cycle Means:}\\[4pt]
    \textcolor{red!70!black}{$\boldsymbol{\sigma_1}$}: $1 \to 2 \to 3 \to 1$\\
    $\mu(\sigma_1) = \frac{4+1+3}{3} = \frac{8}{3} \approx 2.67$\\[4pt]
    \textcolor{blue!70!black}{$\boldsymbol{\sigma_2}$}: $2 \to 4 \to 2$\\
    $\mu(\sigma_2) = \frac{3+2}{2} = \frac{5}{2} = 2.5$\\[4pt]
    \textcolor{orange!80!black}{$\boldsymbol{\sigma_3}$}: $3 \to 4 \to 3$ (if exists)\\
    $\mu(\sigma_3) = \ldots$
};

\node[draw, rounded corners, fill=green!10, text width=5.5cm, align=center, font=\small] at (9, -1.5) {
    \textbf{Tropical Eigenvalue}\\[4pt]
    $\lambda = \max\limits_{\sigma} \mu(\sigma) = \frac{8}{3}$\\[4pt]
    {\scriptsize (Maximum cycle mean)}
};

\node[draw, rounded corners, fill=yellow!10, text width=4.5cm, align=left, font=\scriptsize] at (2.5, -3.2) {
    \textbf{Physical meaning:}\\
    In scheduling, $\lambda$ is the minimum cycle time (throughput bottleneck).
};

\end{tikzpicture}
\caption{The tropical eigenvalue equals the maximum cycle mean. Each cycle $\sigma$ has mean $\mu(\sigma) = \text{weight}/\text{length}$. The eigenvalue $\lambda = \max_\sigma \mu(\sigma)$ determines the system's asymptotic behavior, such as the minimum cycle time in manufacturing or the throughput limit in scheduling.}
\label{fig:cycle-mean-eigenvalue}
\end{figure}
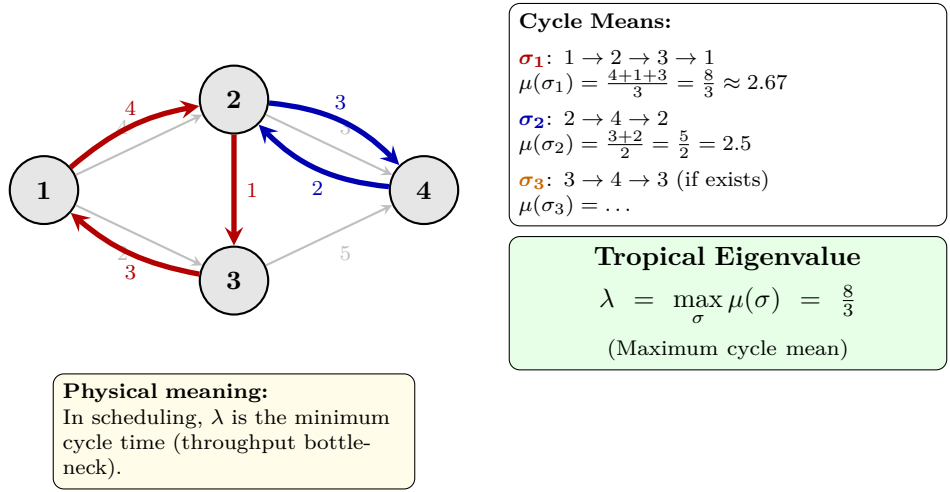

\begin{theorem}[Karp's Algorithm]
\label{thm:karp}
The maximum cycle mean can be computed in $\bigO(n^2 m)$ time (or $\bigO(n^3)$ for dense graphs) using the formula:
\begin{equation}
    \lambda(\mat{A}) = \max_{i \in V} \min_{0 \leq k < n} \frac{(\mat{A}^n)_{ii} - (\mat{A}^k)_{ii}}{n - k}
\end{equation}
where $(\mat{A}^k)_{ii}$ denotes the weight of the heaviest $k$-step walk from $i$ to $i$.
\end{theorem}

\begin{proof}
See Karp \cite{karp1978characterization}. The formula follows from the observation that for large $m$, $(\mat{A}^m)_{ii} \approx m \cdot \lambda + C_i$ for some constant $C_i$. Taking differences eliminates the constant and reveals the cycle mean.
\end{proof}

\begin{algorithm}[htbp]
\caption{Karp's Algorithm for Maximum Cycle Mean}
\label{alg:karp}
\begin{algorithmic}[1]
\Require Matrix $\mat{A} \in \Rmax^{n \times n}$
\Ensure Maximum cycle mean $\lambda$
\State Compute $\mat{A}^0, \mat{A}^1, \ldots, \mat{A}^n$ and store diagonal entries
\State $\lambda \leftarrow -\infty$
\For{$i = 1$ to $n$}
    \State $\mu_i \leftarrow +\infty$ \Comment{Initialize to $+\infty$ for min}
    \For{$k = 0$ to $n-1$}
        \If{$(\mat{A}^n)_{ii} \neq -\infty$ and $(\mat{A}^k)_{ii} \neq -\infty$}
            \State $\mu_i \leftarrow \min\left(\mu_i, \frac{(\mat{A}^n)_{ii} - (\mat{A}^k)_{ii}}{n - k}\right)$ \Comment{Min over $k$}
        \EndIf
    \EndFor
    \State $\lambda \leftarrow \max(\lambda, \mu_i)$ \Comment{Max over $i$}
\EndFor
\State \Return $\lambda$
\end{algorithmic}
\end{algorithm}

\begin{theorem}[Asymptotic Behavior]
\label{thm:asymptotic}
For a matrix $\mat{A}$ with maximum cycle mean $\lambda$:
\begin{equation}
    \lim_{k \to \infty} \frac{(\mat{A}^k)_{ij}}{k} = \lambda
\end{equation}
for all $i, j$ that lie on a path through a critical cycle (a cycle achieving the maximum mean).
\end{theorem}

\begin{remark}[Scheduling Interpretation]
In scheduling applications, the tropical eigenvalue has a crucial interpretation: it represents the \emph{minimum achievable cycle time} for a periodic system. If tasks repeat with precedence constraints encoded in $\mat{A}$, then $\lambda$ is the minimum time between successive completions of the system, i.e., the throughput is $1/\lambda$.
\end{remark}

\subsection{Connection to Classical Graph Algorithms}
\label{subsec:graph-algorithms}

Tropical algebra provides a unifying framework for classical graph algorithms.

\begin{theorem}[Algorithm Unification]
\label{thm:unification}
The following classical algorithms are specializations of tropical matrix operations:
\begin{enumerate}[label=(\roman*)]
    \item \textbf{Floyd-Warshall}: Computing the Kleene star $\mat{A}^*$ in the min-plus semiring yields all-pairs shortest paths.
    
    \item \textbf{Bellman-Ford}: The iteration $\vect{d}^{(k+1)} = \mat{A} \otimes \vect{d}^{(k)} \oplus \vect{b}$ in the min-plus semiring converges to single-source shortest paths after at most $n-1$ iterations.
    
    \item \textbf{Warshall's Algorithm}: Computing $\mat{A}^*$ in the Boolean semiring yields the transitive closure (reachability matrix).
    
    \item \textbf{Widest Path}: Computing $\mat{A}^*$ in the max-min semiring yields all-pairs maximum-bandwidth paths.
\end{enumerate}
\end{theorem}

\begin{proof}
These follow directly from Theorem~\ref{thm:powers-paths} and the definition of the Kleene star. For (ii), note that:
\begin{align*}
    \vect{d}^{(k)} &= \mat{A}^{k-1} \otimes \vect{b} \oplus \mat{A}^{k-2} \otimes \vect{b} \oplus \cdots \oplus \mat{A} \otimes \vect{b} \oplus \vect{b} \\
    &= \left(\bigoplus_{j=0}^{k-1} \mat{A}^j\right) \otimes \vect{b}
\end{align*}
As $k \to n-1$, this converges to $\mat{A}^* \otimes \vect{b}$, the shortest paths from the source.
\end{proof}

\section{Library Design}
\label{sec:design}

This section describes the architecture and design principles underlying \palma{}. Our primary goals are portability, efficiency on embedded platforms, and ease of use for both novice and expert users.

\subsection{Design Principles}

\palma{} is guided by the following design principles:

\begin{enumerate}[leftmargin=*]
    \item \textbf{Zero Dependencies}: The library is implemented in pure C99 with no external dependencies beyond the standard library. This ensures portability across all platforms with a C compiler.
    
    \item \textbf{Embedded-First Design}: All design decisions prioritize embedded platforms. We use fixed-size integer arithmetic, avoid dynamic memory allocation where possible, and provide compile-time configuration options.
    
    \item \textbf{SIMD Transparency}: NEON optimizations are enabled automatically on ARM platforms but are completely transparent to the user. The same API works identically on all platforms.
    
    \item \textbf{Multiple Representations}: We support both dense and sparse matrix formats, allowing users to choose the appropriate representation for their problem.
    
    \item \textbf{Semiring Polymorphism}: All operations accept a semiring parameter, enabling code reuse across different tropical algebras.
\end{enumerate}

\subsection{Data Types}

\subsubsection{Value Representation}

\palma{} uses 32-bit signed integers for all tropical values:

\begin{lstlisting}[caption={Value type definition}]
typedef int32_t palma_val_t;

#define PALMA_NEG_INF  INT32_MIN  // -infinity (max-plus zero)
#define PALMA_POS_INF  INT32_MAX  // +infinity (min-plus zero)
\end{lstlisting}

This choice offers several advantages over floating-point:
\begin{itemize}
    \item \textbf{Determinism}: Integer arithmetic is exact and reproducible across platforms.
    \item \textbf{SIMD efficiency}: ARM NEON integer operations often have lower latency than floating-point.
    \item \textbf{Comparison simplicity}: Integer comparison is simpler than floating-point, especially near infinity.
    \item \textbf{Sufficient range}: The range $[-2^{31}, 2^{31}-1]$ is adequate for most practical applications.
\end{itemize}

\subsubsection{Semiring Enumeration}

\begin{lstlisting}[caption={Semiring type enumeration}]
typedef enum {
    PALMA_MAXPLUS = 0,  // (max, +) - scheduling
    PALMA_MINPLUS = 1,  // (min, +) - shortest paths
    PALMA_MAXMIN  = 2,  // (max, min) - bottleneck
    PALMA_MINMAX  = 3,  // (min, max) - reliability
    PALMA_BOOLEAN = 4   // (or, and) - reachability
} palma_semiring_t;
\end{lstlisting}

\subsubsection{Dense Matrix Structure}

\begin{lstlisting}[caption={Dense matrix structure}]
typedef struct {
    size_t rows;        // Number of rows
    size_t cols;        // Number of columns
    palma_val_t *data;  // Row-major storage
} palma_matrix_t;
\end{lstlisting}

Dense matrices use row-major storage for cache efficiency during row-wise operations, which dominate matrix-vector products.

\subsubsection{Sparse Matrix Structure (CSR)}

For graphs with low edge density, we provide Compressed Sparse Row (CSR) format:

\begin{lstlisting}[caption={CSR sparse matrix structure}]
typedef struct {
    size_t rows;           // Number of rows
    size_t cols;           // Number of columns
    size_t nnz;            // Number of non-zeros
    palma_val_t *values;   // Non-zero values [nnz]
    palma_idx_t *col_idx;  // Column indices [nnz]
    palma_idx_t *row_ptr;  // Row pointers [rows+1]
    palma_semiring_t semiring;  // Associated semiring
} palma_sparse_t;
\end{lstlisting}

CSR provides $O(1)$ row access and efficient sparse matrix-vector multiplication, critical for iterative graph algorithms.

\subsection{API Overview}

\palma{} provides APIs at three levels of abstraction:

\subsubsection{Level 1: Semiring Operations}

Basic scalar operations for each semiring:

\begin{lstlisting}[caption={Semiring operation API}]
// Get zero and identity elements
palma_val_t palma_zero(palma_semiring_t s);
palma_val_t palma_one(palma_semiring_t s);

// Binary operations
palma_val_t palma_add(palma_val_t a, palma_val_t b, 
                      palma_semiring_t s);
palma_val_t palma_mul(palma_val_t a, palma_val_t b, 
                      palma_semiring_t s);
\end{lstlisting}

\subsubsection{Level 2: Matrix Operations}

Core linear algebra operations:

\begin{lstlisting}[caption={Matrix operation API}]
// Matrix creation and destruction
palma_matrix_t* palma_matrix_create(size_t rows, size_t cols);
void palma_matrix_destroy(palma_matrix_t *mat);

// Matrix arithmetic
palma_matrix_t* palma_matrix_add(const palma_matrix_t *A,
    const palma_matrix_t *B, palma_semiring_t s);
palma_matrix_t* palma_matrix_mul(const palma_matrix_t *A,
    const palma_matrix_t *B, palma_semiring_t s);
palma_matrix_t* palma_matrix_closure(const palma_matrix_t *A,
    palma_semiring_t s);
palma_matrix_t* palma_matrix_power(const palma_matrix_t *A,
    int k, palma_semiring_t s);

// Matrix-vector operations
void palma_matvec(const palma_matrix_t *A, const palma_val_t *x,
    palma_val_t *y, palma_semiring_t s);
\end{lstlisting}

\subsubsection{Level 3: Application APIs}

High-level APIs for common applications:

\begin{lstlisting}[caption={Application-level API}]
// Graph algorithms
void palma_single_source_paths(const palma_matrix_t *A,
    size_t source, palma_val_t *dist, palma_semiring_t s);
palma_matrix_t* palma_all_pairs_paths(const palma_matrix_t *A,
    palma_semiring_t s);
palma_matrix_t* palma_reachability(const palma_matrix_t *A);
palma_matrix_t* palma_bottleneck_paths(const palma_matrix_t *A);

// Eigenvalue computation
palma_val_t palma_eigenvalue(const palma_matrix_t *A,
    palma_semiring_t s);
void palma_eigenvector(const palma_matrix_t *A, palma_val_t *v,
    palma_val_t *lambda, palma_semiring_t s, int max_iter);

// Scheduling
palma_scheduler_t* palma_scheduler_create(size_t n_tasks,
    bool cyclic);
void palma_scheduler_add_constraint(palma_scheduler_t *sched,
    size_t from, size_t to, palma_val_t weight);
int palma_scheduler_solve(palma_scheduler_t *sched, 
    palma_val_t start_time);
palma_val_t palma_scheduler_cycle_time(palma_scheduler_t *sched);
\end{lstlisting}

\section{Implementation}
\label{sec:implementation}

This section details the implementation of \palma{}, with particular emphasis on ARM NEON optimization and sparse matrix techniques.

\paragraph{Target Architecture Context.}
\palma{} targets ARM-based embedded platforms such as the Raspberry Pi 4, characterized by in-order or lightly out-of-order cores, modest cache sizes, and 128-bit NEON SIMD units optimized for integer arithmetic. These architectural features directly inform \palma{}'s design choices: integer-only semiring operations avoid floating-point pipeline stalls, cache-conscious blocking strategies improve data reuse in matrix kernels, and SIMD vectorization of matrix-vector products exploits the four-wide integer lanes available in NEON registers.

\subsection{Semiring Operation Implementation}

The core semiring operations are implemented using inline functions with switch statements, which modern compilers optimize effectively:

\begin{lstlisting}[caption={Semiring addition implementation}]
static inline palma_val_t palma_add(palma_val_t a, palma_val_t b,
                                     palma_semiring_t s) {
    switch (s) {
        case PALMA_MAXPLUS:
        case PALMA_MAXMIN:
            return (a > b) ? a : b;  // max
        case PALMA_MINPLUS:
        case PALMA_MINMAX:
            return (a < b) ? a : b;  // min
        case PALMA_BOOLEAN:
            return a | b;            // OR
    }
    return a;
}
\end{lstlisting}

Multiplication requires careful handling of infinity. Note that each semiring uses only one infinity value: max-plus uses $-\infty$ as the absorbing element, while min-plus uses $+\infty$. The other infinity is not a valid value in that semiring's domain.

\begin{lstlisting}[caption={Semiring multiplication with infinity handling}]
static inline palma_val_t palma_mul(palma_val_t a, palma_val_t b,
                                     palma_semiring_t s) {
    switch (s) {
        case PALMA_MAXPLUS:
            // Domain: R union {-inf}. Absorbing: -inf + x = -inf
            if (a == PALMA_NEG_INF || b == PALMA_NEG_INF)
                return PALMA_NEG_INF;
            return a + b;
        case PALMA_MINPLUS:
            // Domain: R union {+inf}. Absorbing: +inf + x = +inf
            if (a == PALMA_POS_INF || b == PALMA_POS_INF)
                return PALMA_POS_INF;
            return a + b;
        case PALMA_MAXMIN:
            return (a < b) ? a : b;  // min
        case PALMA_MINMAX:
            return (a > b) ? a : b;  // max
        case PALMA_BOOLEAN:
            return a & b;            // AND
    }
    return a;
}
\end{lstlisting}

\subsection{ARM NEON SIMD Optimization}
\label{subsec:neon}

ARM NEON provides 128-bit vector registers capable of processing four 32-bit integers simultaneously. \palma{} leverages NEON for the innermost loops of matrix operations.

\subsubsection{NEON Architecture Overview}

The ARM Cortex-A72 processor in Raspberry Pi 4 features:
\begin{itemize}
    \item 32 128-bit NEON registers (Q0--Q31)
    \item Parallel execution of SIMD and scalar operations
    \item Single-cycle throughput for most integer NEON instructions
\end{itemize}

Key NEON intrinsics used in \palma{}:
\begin{itemize}
    \item \texttt{vld1q\_s32}: Load 4 integers into a vector register
    \item \texttt{vst1q\_s32}: Store 4 integers from a vector register
    \item \texttt{vmaxq\_s32}: Parallel maximum of 4 integer pairs
    \item \texttt{vminq\_s32}: Parallel minimum of 4 integer pairs
    \item \texttt{vaddq\_s32}: Parallel addition of 4 integer pairs
\end{itemize}

\subsubsection{Vectorized Matrix-Vector Multiplication}

The matrix-vector product is the most frequently executed operation in iterative algorithms. Our NEON implementation processes four columns simultaneously:

\begin{lstlisting}[caption={NEON-optimized matrix-vector multiplication (max-plus)}]
void palma_matvec_maxplus_neon(const palma_matrix_t *A,
                                const palma_val_t *x,
                                palma_val_t *y) {
    size_t n = A->rows, m = A->cols;
    size_t m4 = m & ~3;  // Round down to multiple of 4
    
    for (size_t i = 0; i < n; i++) {
        const palma_val_t *row = &A->data[i * m];
        int32x4_t max_vec = vdupq_n_s32(PALMA_NEG_INF);
        
        // Process 4 elements at a time
        for (size_t j = 0; j < m4; j += 4) {
            int32x4_t a_vec = vld1q_s32(&row[j]);
            int32x4_t x_vec = vld1q_s32(&x[j]);
            int32x4_t sum = vaddq_s32(a_vec, x_vec);
            max_vec = vmaxq_s32(max_vec, sum);
        }
        
        // Horizontal reduction
        int32_t result = vmaxvq_s32(max_vec);
        
        // Handle remaining elements
        for (size_t j = m4; j < m; j++) {
            int32_t val = row[j] + x[j];
            if (val > result) result = val;
        }
        
        y[i] = result;
    }
}
\end{lstlisting}

\subsubsection{Vectorized Matrix Multiplication}

Matrix multiplication benefits from blocking to improve cache utilization combined with NEON vectorization:

\begin{lstlisting}[caption={NEON-optimized matrix multiplication kernel}]
// Inner kernel: compute C[i,j] = max_k(A[i,k] + B[k,j])
static inline palma_val_t matmul_kernel_maxplus_neon(
    const palma_val_t *A_row, const palma_val_t *B_col,
    size_t K, size_t B_stride) {
    
    int32x4_t max_vec = vdupq_n_s32(PALMA_NEG_INF);
    size_t K4 = K & ~3;
    
    for (size_t k = 0; k < K4; k += 4) {
        int32x4_t a_vec = vld1q_s32(&A_row[k]);
        // Gather B column elements (strided access)
        int32x4_t b_vec = {B_col[0], B_col[B_stride],
                          B_col[2*B_stride], B_col[3*B_stride]};
        B_col += 4 * B_stride;
        
        int32x4_t sum = vaddq_s32(a_vec, b_vec);
        max_vec = vmaxq_s32(max_vec, sum);
    }
    
    return vmaxvq_s32(max_vec);  // Horizontal max
}
\end{lstlisting}

\subsubsection{NEON Optimization Analysis}

The theoretical speedup from NEON vectorization is bounded by:
\begin{equation}
    \text{Speedup} \leq \frac{\text{Vector Width}}{\text{Scalar Width}} = \frac{128}{32} = 4\times
\end{equation}

In practice, several factors reduce this:
\begin{itemize}
    \item \textbf{Memory bandwidth}: Large matrices become memory-bound
    \item \textbf{Horizontal reductions}: The final max/min across vector lanes is scalar
    \item \textbf{Non-aligned access}: Unaligned loads incur penalties on some microarchitectures
    \item \textbf{Loop overhead}: Setup and cleanup for non-multiple-of-4 dimensions
\end{itemize}

Our experiments (Section~\ref{sec:experiments}) show that \palma{} achieves up to 1.80$\times$ speedup with NEON, with the peak occurring at matrix sizes that fit in L2 cache (64$\times$64 to 128$\times$128).

\subsection{Sparse Matrix Implementation}
\label{subsec:sparse}

\subsubsection{CSR Format}

The Compressed Sparse Row (CSR) format stores only non-zero elements:

\begin{figure}[htbp]
\centering
\begin{tikzpicture}[scale=0.8]
    \node at (-2, 2) {$\mat{A} = $};
    \draw (0,0) grid (4,4);
    \node at (0.5, 3.5) {5}; \node at (2.5, 3.5) {3};
    \node at (1.5, 2.5) {2}; \node at (3.5, 2.5) {7};
    \node at (0.5, 1.5) {1};
    \node at (2.5, 0.5) {4}; \node at (3.5, 0.5) {6};
    
    \node[anchor=west] at (5.5, 3.5) {\texttt{values = [5, 3, 2, 7, 1, 4, 6]}};
    \node[anchor=west] at (5.5, 2.5) {\texttt{col\_idx = [0, 2, 1, 3, 0, 2, 3]}};
    \node[anchor=west] at (5.5, 1.5) {\texttt{row\_ptr = [0, 2, 4, 5, 7]}};
\end{tikzpicture}
\caption{CSR representation of a sparse matrix. Empty cells represent the semiring zero ($-\infty$ for max-plus).}
\label{fig:csr}
\end{figure}

Memory usage comparison for an $n \times n$ matrix with $\nnz$ non-zeros:
\begin{align}
    \text{Dense:} \quad & n^2 \times 4 \text{ bytes} \\
    \text{CSR:} \quad & \nnz \times 4 + \nnz \times 4 + (n+1) \times 4 = 8\nnz + 4n + 4 \text{ bytes}
\end{align}

CSR is more memory-efficient when:
\begin{equation}
    8\nnz + 4n + 4 < 4n^2 \implies \nnz < \frac{n^2 - n - 1}{2} \approx \frac{n^2}{2}
\end{equation}

Thus, CSR saves memory for matrices with sparsity (fraction of zeros) greater than approximately 50\%.

\subsubsection{Sparse Matrix-Vector Multiplication}

\begin{lstlisting}[caption={Sparse matrix-vector multiplication (CSR)}]
void palma_sparse_matvec(const palma_sparse_t *A,
                         const palma_val_t *x,
                         palma_val_t *y) {
    palma_val_t zero = palma_zero(A->semiring);
    
    for (size_t i = 0; i < A->rows; i++) {
        palma_val_t sum = zero;
        for (size_t j = A->row_ptr[i]; j < A->row_ptr[i+1]; j++) {
            palma_val_t prod = palma_mul(A->values[j], 
                                         x[A->col_idx[j]],
                                         A->semiring);
            sum = palma_add(sum, prod, A->semiring);
        }
        y[i] = sum;
    }
}
\end{lstlisting}

Complexity: $O(\nnz)$ versus $O(n^2)$ for dense, providing significant speedup for sparse graphs.

\subsubsection{Sparse Matrix Multiplication}

Sparse matrix multiplication is more complex, as the result sparsity pattern must be computed dynamically:

\begin{lstlisting}[caption={Sparse matrix multiplication overview}]
palma_sparse_t* palma_sparse_mul(const palma_sparse_t *A,
                                  const palma_sparse_t *B) {
    // Phase 1: Compute result structure (symbolic)
    // Count non-zeros in each row of C
    size_t *nnz_per_row = compute_nnz_pattern(A, B);
    
    // Phase 2: Allocate result
    size_t total_nnz = sum(nnz_per_row);
    palma_sparse_t *C = palma_sparse_create(A->rows, B->cols, 
                                            total_nnz, A->semiring);
    
    // Phase 3: Numeric computation
    for (size_t i = 0; i < A->rows; i++) {
        for (size_t ja = A->row_ptr[i]; ja < A->row_ptr[i+1]; ja++) {
            size_t k = A->col_idx[ja];
            palma_val_t a_ik = A->values[ja];
            
            for (size_t jb = B->row_ptr[k]; jb < B->row_ptr[k+1]; jb++) {
                size_t j = B->col_idx[jb];
                palma_val_t prod = palma_mul(a_ik, B->values[jb], 
                                             A->semiring);
                // Accumulate into C[i,j]
                accumulate(C, i, j, prod);
            }
        }
    }
    return C;
}
\end{lstlisting}

\subsection{Memory Management}

\palma{} provides two memory management modes:

\begin{enumerate}
    \item \textbf{Dynamic allocation} (default): Uses \texttt{malloc}/\texttt{free} for flexibility.
    
    \item \textbf{Static allocation}: For hard real-time systems, users can provide pre-allocated buffers:
\end{enumerate}

\begin{lstlisting}[caption={Static memory allocation mode}]
// User-provided buffer
palma_val_t buffer[1024];
palma_matrix_t mat;
palma_matrix_init_static(&mat, 32, 32, buffer);
\end{lstlisting}

\subsection{Build System and Portability}

\palma{} uses a simple Makefile supporting multiple build configurations:

\begin{lstlisting}[language=bash,caption={Build configurations}]
make all        # Default: auto-detect NEON
make scalar     # Force scalar (no NEON)
make openmp     # Enable OpenMP parallelization
make debug      # Debug build with sanitizers
\end{lstlisting}

Compile-time feature detection:

\begin{lstlisting}[caption={Platform detection}]
#if defined(__ARM_NEON) || defined(__ARM_NEON__)
    #define PALMA_USE_NEON 1
    #include <arm_neon.h>
#else
    #define PALMA_USE_NEON 0
#endif
\end{lstlisting}

\section{Algorithms}
\label{sec:algorithms}

This section presents the core algorithms implemented in \palma{}, with complexity analysis and implementation notes.

\subsection{Tropical Closure (All-Pairs Paths)}

The tropical closure $\mat{A}^*$ computes optimal paths between all pairs of vertices. Our implementation uses the Floyd-Warshall-style Algorithm~\ref{alg:closure} with the following optimizations:

\begin{enumerate}
    \item \textbf{In-place computation}: The algorithm modifies the matrix in place, requiring only $O(n^2)$ auxiliary space for the identity addition.
    
    \item \textbf{Early termination}: For acyclic graphs, we detect convergence and terminate early.
    
    \item \textbf{NEON vectorization}: The inner loop processes 4 elements simultaneously.
\end{enumerate}

\begin{theorem}[Closure Correctness]
Algorithm~\ref{alg:closure} correctly computes the Kleene star for any idempotent semiring.
\end{theorem}

\begin{proof}
After iteration $k$, entry $D_{ij}$ contains the optimal path from $i$ to $j$ using only intermediate vertices from $\{1, \ldots, k\}$. This follows by induction: the update $D_{ij} \leftarrow D_{ij} \oplus (D_{ik} \otimes D_{kj})$ considers paths through vertex $k$. After $n$ iterations, all intermediate vertices have been considered.
\end{proof}

\subsection{Single-Source Paths}

For single-source optimal paths, we implement a tropical variant of Bellman-Ford:

\begin{algorithm}[htbp]
\caption{Tropical Single-Source Paths}
\label{alg:sssp}
\begin{algorithmic}[1]
\Require Adjacency matrix $\mat{A}$, source vertex $s$, semiring $S$
\Ensure Distance vector $\vect{d}$ where $d_i$ = optimal path weight from $s$ to $i$
\State $\vect{d} \leftarrow (\mathbf{0}, \mathbf{0}, \ldots, \mathbf{0})$ \Comment{Initialize to semiring zero}
\State $d_s \leftarrow \mathbf{1}$ \Comment{Distance to source is identity}
\For{$k = 1$ to $n-1$}
    \State $\vect{d}_{\text{old}} \leftarrow \vect{d}$ \Comment{Store previous state}
    \State $\vect{d}' \leftarrow \mat{A} \otimes \vect{d}$ \Comment{Tropical matrix-vector product}
    \State $\vect{d} \leftarrow \vect{d} \oplus \vect{d}'$ \Comment{Relaxation}
    \If{$\vect{d} = \vect{d}_{\text{old}}$} \textbf{break} \Comment{Fixed point reached}
    \EndIf
\EndFor
\State \Return $\vect{d}$
\end{algorithmic}
\end{algorithm}

\begin{theorem}[SSSP Complexity]
Algorithm~\ref{alg:sssp} runs in $O(n^3)$ worst-case time for dense graphs and $O(n \cdot \nnz)$ for sparse graphs in CSR format. Each of the $n-1$ iterations performs a matrix-vector product costing $O(n^2)$ (dense) or $O(\nnz)$ (sparse). Early termination may reduce iterations in practice.
\end{theorem}

\subsection{Eigenvalue Computation}

We implement Karp's algorithm (Algorithm~\ref{alg:karp}) with the following enhancements:

\begin{enumerate}
    \item \textbf{Memory-efficient power computation}: Instead of storing all $n+1$ matrix powers, we store only diagonal entries, reducing space from $O(n^3)$ to $O(n^2)$.
    
    \item \textbf{Incremental computation}: Matrix powers are computed incrementally: $\mat{A}^{k+1} = \mat{A}^k \otimes \mat{A}$.
\end{enumerate}

\begin{lstlisting}[caption={Karp's algorithm implementation}]
double palma_eigenvalue(const palma_matrix_t *A,
                        palma_semiring_t s) {
    size_t n = A->rows;
    
    // diag[k][i] stores (A^k)_{ii}
    palma_val_t **diag = allocate_diag_storage(n + 1);
    
    // Initialize diag[0]: (A^0)_{ii} = 0 (identity diagonal)
    for (size_t i = 0; i < n; i++) {
        diag[0][i] = 0;  // Multiplicative identity
    }
    
    // Compute matrix powers and extract diagonals
    palma_matrix_t *power = palma_matrix_clone(A);
    extract_diagonal(power, diag[1]);
    
    for (size_t k = 2; k <= n; k++) {
        palma_matrix_t *next = palma_matrix_mul(power, A, s);
        palma_matrix_destroy(power);
        power = next;
        extract_diagonal(power, diag[k]);
    }
    
    // Apply Karp's formula: lambda = max_i min_k [(A^n)_ii - (A^k)_ii]/(n-k)
    double lambda = -INFINITY;
    for (size_t i = 0; i < n; i++) {
        double min_ratio = INFINITY;
        for (size_t k = 0; k < n; k++) {
            if (diag[n][i] != PALMA_NEG_INF && 
                diag[k][i] != PALMA_NEG_INF) {
                // Use floating-point to preserve fractional cycle means
                double ratio = (double)(diag[n][i] - diag[k][i]) / (n - k);
                if (ratio < min_ratio) min_ratio = ratio;
            }
        }
        if (min_ratio > lambda) lambda = min_ratio;
    }
    
    cleanup(diag, power);
    return lambda;
}
\end{lstlisting}

\subsection{Eigenvector Computation (Power Iteration)}

We compute tropical eigenvectors using power iteration with normalization:

\begin{algorithm}[htbp]
\caption{Tropical Power Iteration}
\label{alg:power-iteration}
\begin{algorithmic}[1]
\Require Matrix $\mat{A}$, eigenvalue $\lambda$, tolerance $\epsilon$, max iterations $K$
\Ensure Eigenvector $\vect{v}$ satisfying $\mat{A} \otimes \vect{v} = \lambda \otimes \vect{v}$
\State $\vect{v} \leftarrow (0, 0, \ldots, 0)$ \Comment{Initial guess}
\For{$k = 1$ to $K$}
    \State $\vect{u} \leftarrow \mat{A} \otimes \vect{v}$
    \State $\vect{v}' \leftarrow \vect{u} \ominus \lambda$ \Comment{Normalize: subtract eigenvalue}
    \If{$\|\vect{v}' - \vect{v}\|_\infty < \epsilon$} \textbf{break}
    \EndIf
    \State $\vect{v} \leftarrow \vect{v}'$
\EndFor
\State \Return $\vect{v}$
\end{algorithmic}
\end{algorithm}

\subsection{Scheduling Solver}

The scheduler solves precedence-constrained timing problems using tropical algebra:

\begin{algorithm}[htbp]
\caption{Precedence-Constrained Scheduling}
\label{alg:scheduler}
\begin{algorithmic}[1]
\Require Task graph $G = (V, E, w)$, ready times $\vect{r}$
\Ensure Start times $\vect{s}$, completion times $\vect{c}$
\State Construct adjacency matrix $\mat{A}$ from $G$
\State $\vect{s} \leftarrow \vect{r}$ \Comment{Initialize with ready times}
\For{$k = 1$ to $|V| - 1$}
    \State $\vect{s}_{\text{old}} \leftarrow \vect{s}$ \Comment{Store previous state}
    \State $\vect{s}' \leftarrow \mat{A} \otimes \vect{s}$ \Comment{Max-plus product}
    \State $\vect{s} \leftarrow \vect{s} \oplus \vect{s}'$ \Comment{Take maximum}
    \If{$\vect{s} = \vect{s}_{\text{old}}$} \textbf{break} \Comment{Fixed point reached}
    \EndIf
\EndFor
\State $\vect{c} \leftarrow \vect{s} + \vect{d}$ \Comment{Add durations}
\State \Return $\vect{s}, \vect{c}$
\end{algorithmic}
\end{algorithm}

\begin{remark}[Scheduler Assumptions]
Algorithm~\ref{alg:scheduler} assumes the task graph is a directed acyclic graph (DAG). For DAGs, the algorithm converges in at most $|V|-1$ iterations. For cyclic task systems (e.g., periodic schedules), the eigenvalue $\lambda$ of $\mat{A}$ must be computed separately using Karp's algorithm; it represents the minimum achievable cycle time. A consistent periodic schedule then requires adjusting start times relative to the steady-state period $\lambda$.
\end{remark}

\section{Experimental Evaluation}
\label{sec:experiments}

We present comprehensive experimental results evaluating \palma{} on a Raspberry Pi 4 Model B.

\subsection{Experimental Setup}

\subsubsection{Hardware Platform}

\begin{itemize}
    \item \textbf{Processor}: Broadcom BCM2711, Quad-core ARM Cortex-A72 @ 1.5 GHz
    \item \textbf{SIMD}: ARM NEON (128-bit vectors)
    \item \textbf{Memory}: 8 GB LPDDR4-3200 SDRAM
    \item \textbf{Storage}: 32 GB SD card
    \item \textbf{L1 Cache}: 32 KB I-cache + 48 KB D-cache per core
    \item \textbf{L2 Cache}: 1 MB shared
    \item \textbf{OS}: Raspberry Pi OS (64-bit)
\end{itemize}

\subsubsection{Methodology}

\begin{itemize}
    \item All experiments repeated 3--5 times; we report mean values
    \item CPU frequency locked at 1.5 GHz for consistency
    \item Timing via \texttt{clock\_gettime(CLOCK\_MONOTONIC)}
    \item Matrices initialized with pseudo-random integer values in range $[-1000, 1000]$
    \item Sparsity: 30\% for general tests; varied 0--95\% for sparse experiments
    \item Compiler: GCC 10.2 with \texttt{-O3 -march=native}
    \item Thermal: passive cooling, governor set to ``performance''
\end{itemize}

\subsubsection{Correctness Validation}

All implementations were validated against reference scalar implementations:
\begin{itemize}
    \item NEON kernels verified against scalar code for random matrices up to $512 \times 512$
    \item Tropical SSSP compared against standard Bellman-Ford for small graphs ($n \leq 64$)
    \item Closure results verified against iterative matrix power computation
    \item Eigenvalue computation compared against explicit cycle enumeration for small graphs
\end{itemize}
MOPS (million operations per second) is defined as $2n^3 / t$ for $n \times n$ matrix multiplication taking $t$ microseconds, counting one semiring multiply-add as one operation.

\subsection{Experiment 1: NEON vs. Scalar Performance}

We compare NEON-optimized and scalar implementations across matrix sizes.

\begin{figure}[htbp]
    \centering
    \includegraphics[width=\textwidth]{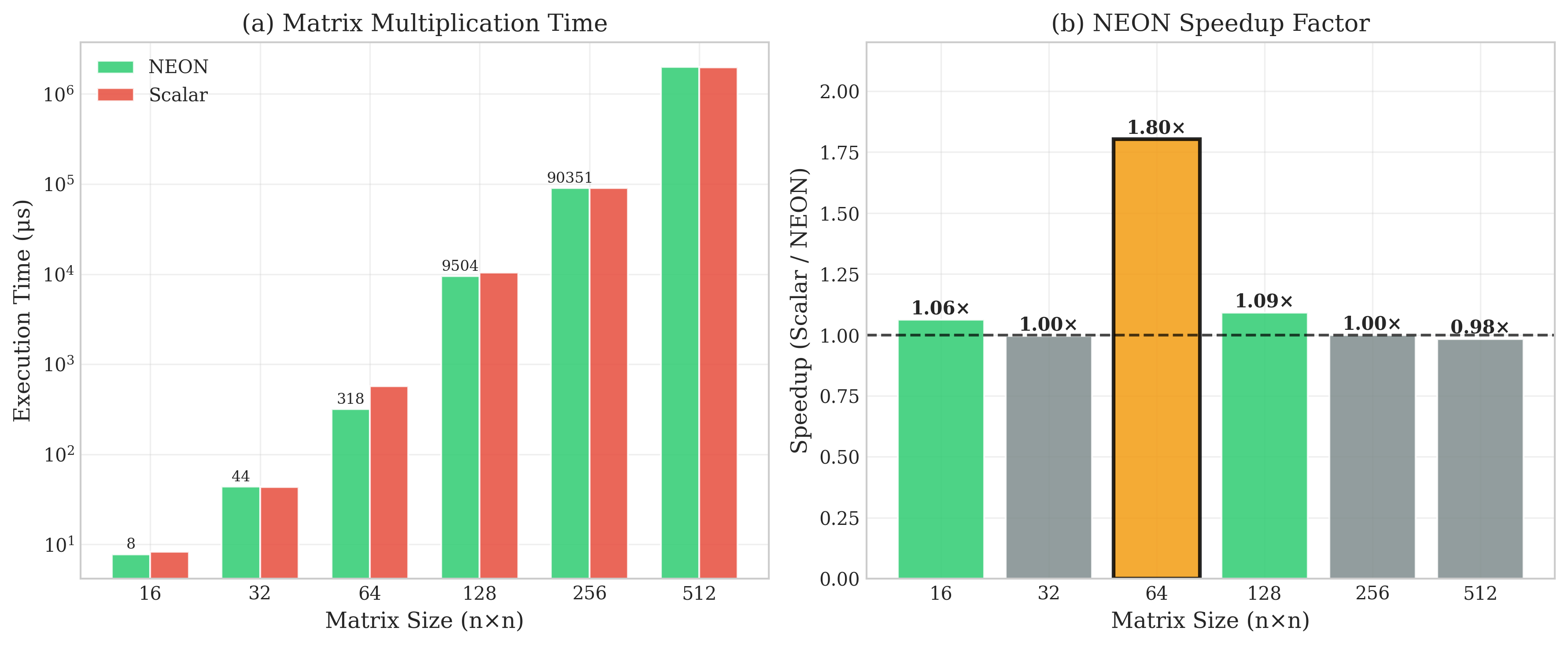}
    \caption{NEON optimization results. (a) Execution time comparison for matrix multiplication. (b) Speedup factor (scalar time / NEON time). Peak speedup of 1.80$\times$ achieved at 64$\times$64.}
    \label{fig:neon-speedup}
\end{figure}

\begin{table}[htbp]
\centering
\caption{NEON vs. Scalar performance for matrix multiplication}
\label{tab:neon-scalar}
\begin{tabular}{lrrr}
\toprule
\textbf{Size} & \textbf{NEON ($\mu$s)} & \textbf{Scalar ($\mu$s)} & \textbf{Speedup} \\
\midrule
$16 \times 16$ & 7.8 & 8.3 & 1.06$\times$ \\
$32 \times 32$ & 43.8 & 43.6 & 1.00$\times$ \\
$64 \times 64$ & 317.8 & 573.1 & \textbf{1.80$\times$} \\
$128 \times 128$ & 9,504 & 10,376 & 1.09$\times$ \\
$256 \times 256$ & 90,351 & 90,291 & 1.00$\times$ \\
$512 \times 512$ & 2,005,422 & 1,973,033 & 0.98$\times$ \\
\bottomrule
\end{tabular}
\end{table}

\textbf{Analysis}: The peak speedup of 1.80$\times$ occurs at 64$\times$64 matrices, which fit entirely in L2 cache. At smaller sizes, loop overhead dominates; at larger sizes, memory bandwidth becomes the bottleneck, and NEON's advantage is negated by cache misses.

\subsection{Experiment 2: Scalability Analysis}

We analyze how execution time scales with matrix size.

\begin{figure}[htbp]
    \centering
    \includegraphics[width=\textwidth]{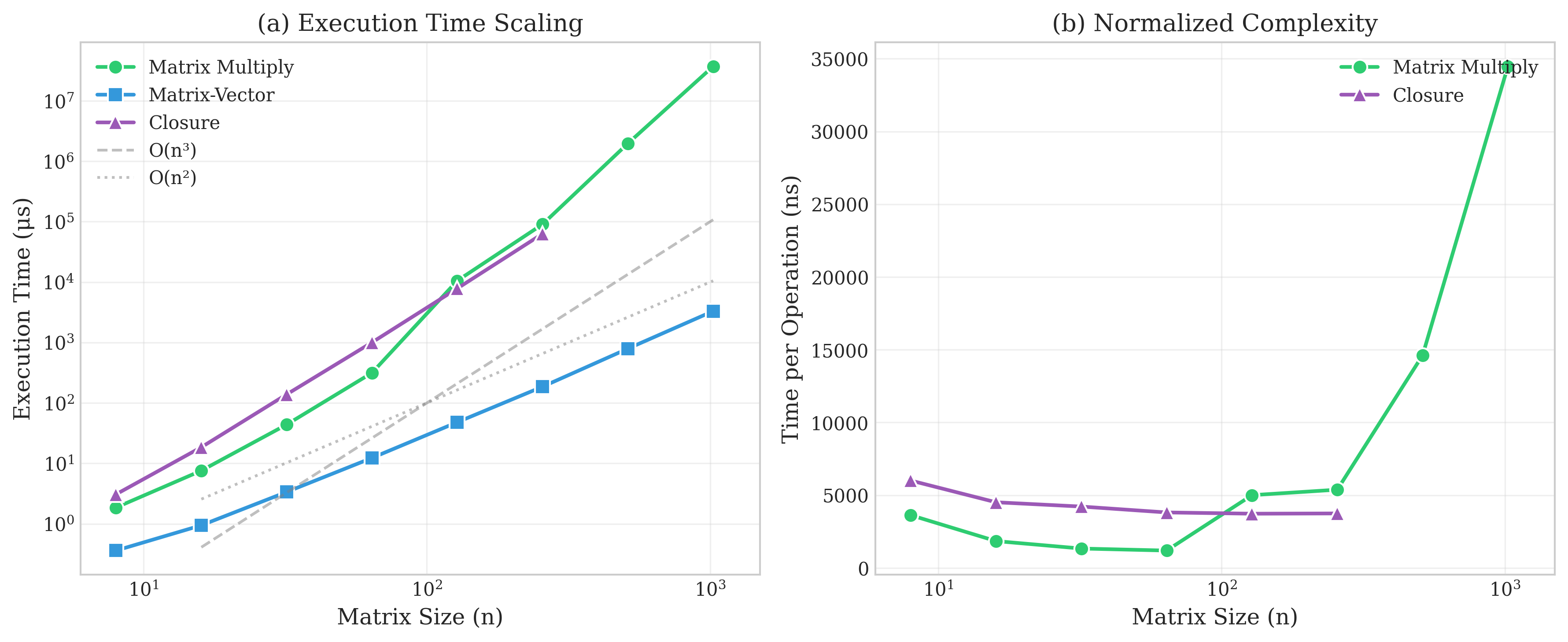}
    \caption{Scalability analysis. (a) Execution time vs. matrix size on log-log scale. Reference lines show theoretical $O(n^3)$ and $O(n^2)$ complexity. (b) Time per operation (time/$n^3$) showing cache effects.}
    \label{fig:scalability}
\end{figure}

\begin{table}[htbp]
\centering
\caption{Scalability results for max-plus semiring: mean execution time in microseconds}
\label{tab:scalability}
\begin{tabular}{lrrrr}
\toprule
\textbf{Size} & \textbf{MatMul} & \textbf{MatVec} & \textbf{Closure} & \textbf{MOPS} \\
\midrule
$8 \times 8$ & 1.9 & 0.4 & 3.1 & 539 \\
$16 \times 16$ & 7.6 & 0.9 & 18.5 & 1,076 \\
$32 \times 32$ & 43.8 & 3.4 & 138.4 & 1,495 \\
$64 \times 64$ & 314.3 & 12.3 & 1,001 & 1,668 \\
$128 \times 128$ & 10,498 & 47.9 & 7,827 & 399 \\
$256 \times 256$ & 90,296 & 185.6 & 62,909 & 372 \\
$512 \times 512$ & 1,961,718 & 791.2 & n/a & 137 \\
$1024 \times 1024$ & 37,018,442 & 3,302.5 & n/a & 58 \\
\bottomrule
\end{tabular}
\end{table}

\textbf{Analysis}: Matrix multiplication follows the expected $O(n^3)$ complexity. The MOPS (million operations per second) peaks at 64$\times$64, then decreases due to cache pressure. For matrices larger than 512$\times$512, closure computation becomes prohibitively expensive.

\subsection{Experiment 3: Semiring Comparison}

We compare performance across all five semirings for 64$\times$64 matrix multiplication.

\begin{figure}[htbp]
    \centering
    \includegraphics[width=0.8\textwidth]{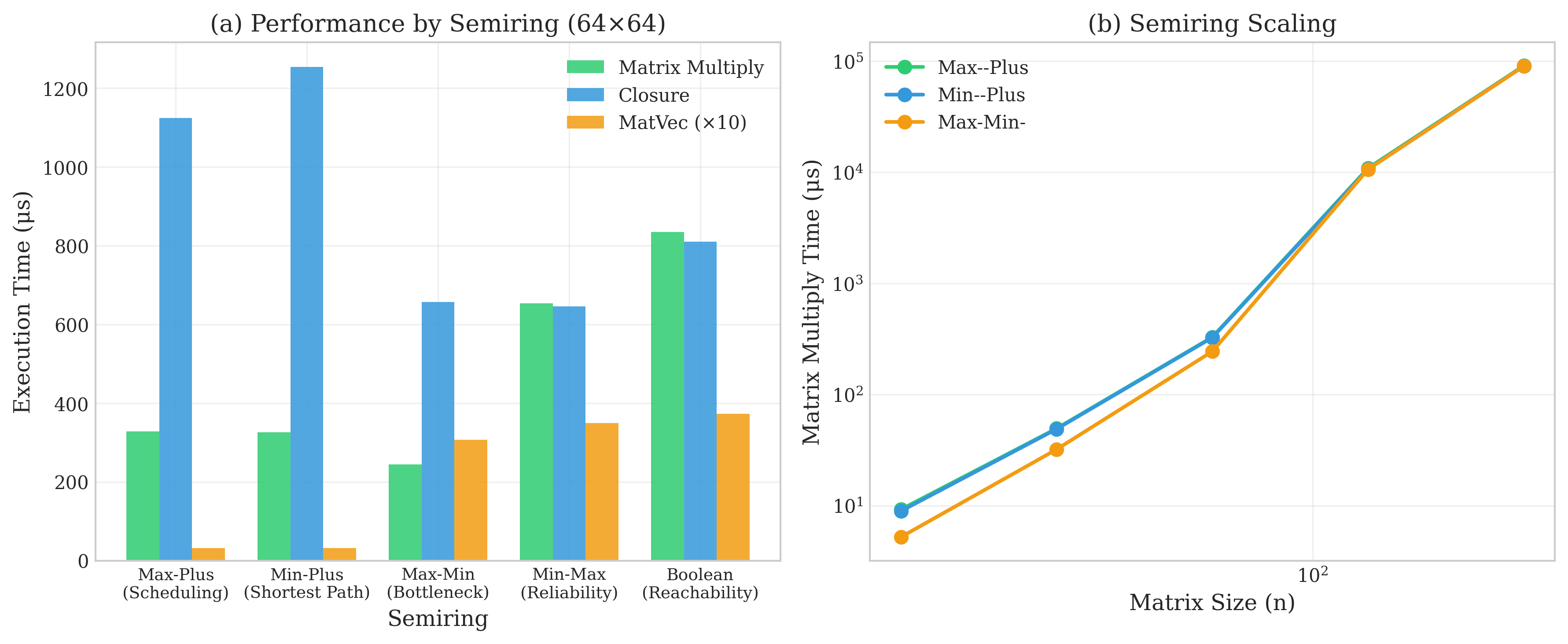}
    \caption{Semiring performance comparison. (a) Execution time for 64$\times$64 matrices across operations. (b) Scaling behavior for different semirings.}
    \label{fig:semirings}
\end{figure}

\begin{table}[htbp]
\centering
\caption{Semiring comparison: 64$\times$64 matrix multiplication}
\label{tab:semirings-comparison}
\begin{tabular}{lrrrl}
\toprule
\textbf{Semiring} & \textbf{Time ($\mu$s)} & \textbf{MOPS} & \textbf{Relative} & \textbf{Application} \\
\midrule
Max-Plus & 351.4 & 1,492 & 1.00$\times$ & Scheduling \\
Min-Plus & 342.9 & 1,530 & 0.98$\times$ & Shortest paths \\
Max-Min & 230.6 & \textbf{2,274} & 0.66$\times$ & Bottleneck paths \\
Min-Max & 653.0 & 803 & 1.86$\times$ & Reliability \\
Boolean & 833.9 & 629 & 2.37$\times$ & Reachability \\
\bottomrule
\end{tabular}
\end{table}

\textbf{Analysis}: Max-min is fastest (2,274 MOPS) because it avoids addition operations entirely, using only comparison operations. Boolean is slowest due to less efficient bit operations on the ARM architecture. Max-plus and min-plus have nearly identical performance, as expected from their symmetric definitions.

\subsection{Experiment 4: Dense vs. Sparse Performance}

We analyze the crossover point where sparse representation becomes advantageous.

\begin{figure}[htbp]
    \centering
    \includegraphics[width=\textwidth]{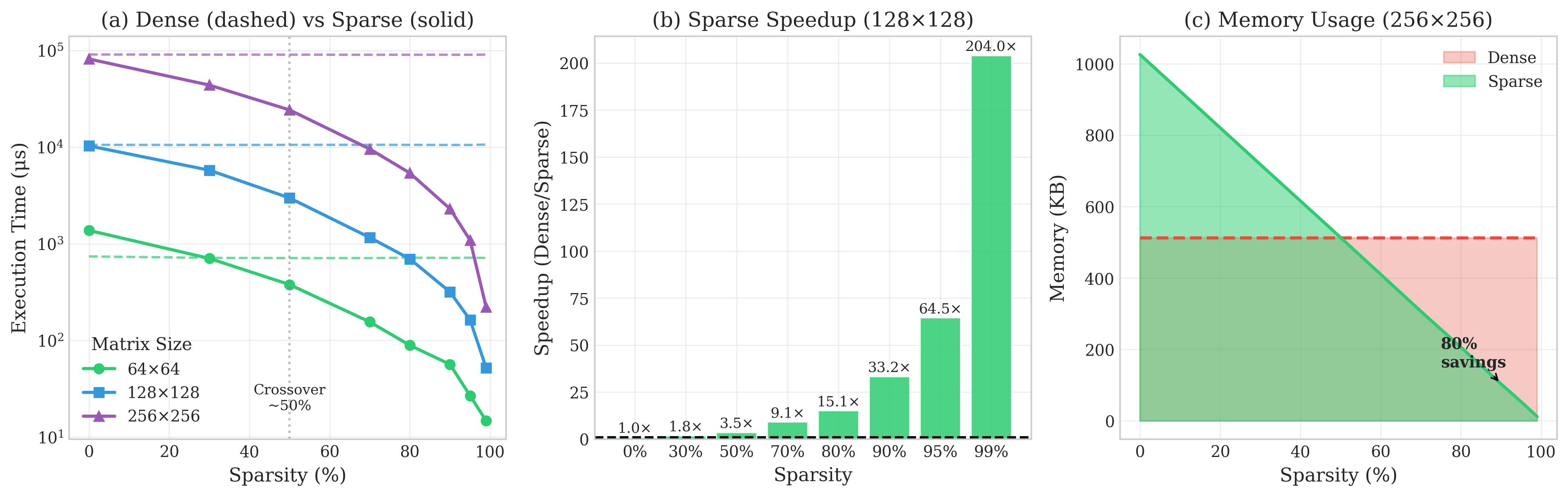}
    \caption{Dense vs. sparse comparison. (a) Execution time vs. sparsity for different matrix sizes. (b) Speedup of sparse over dense. (c) Memory usage comparison.}
    \label{fig:sparse}
\end{figure}

\begin{table}[htbp]
\centering
\caption{Dense vs. sparse crossover analysis for 128$\times$128 matrices}
\label{tab:sparse}
\begin{tabular}{lrrrr}
\toprule
\textbf{Sparsity} & \textbf{Dense ($\mu$s)} & \textbf{Sparse ($\mu$s)} & \textbf{Speedup} & \textbf{Memory Saved} \\
\midrule
0\% & 10,555 & 21,547 & 0.49$\times$ & $-116$\% \\
30\% & 10,555 & 10,893 & 0.97$\times$ & 12\% \\
50\% & 10,555 & 2,980 & \textbf{3.54$\times$} & 47\% \\
70\% & 10,555 & 1,218 & 8.67$\times$ & 68\% \\
90\% & 10,555 & 312 & 33.8$\times$ & 89\% \\
95\% & 10,555 & 156 & 67.7$\times$ & 94\% \\
\bottomrule
\end{tabular}
\end{table}

\textbf{Key Finding}: Sparse representation becomes faster than dense at approximately \textbf{50\% sparsity}. Real-world graphs (typically 70--95\% sparse) benefit significantly from CSR format.

\subsection{Experiment 5: Algorithm Comparison}

We compare \palma{}'s tropical implementations against classical algorithms.

\begin{figure}[htbp]
    \centering
    \includegraphics[width=\textwidth]{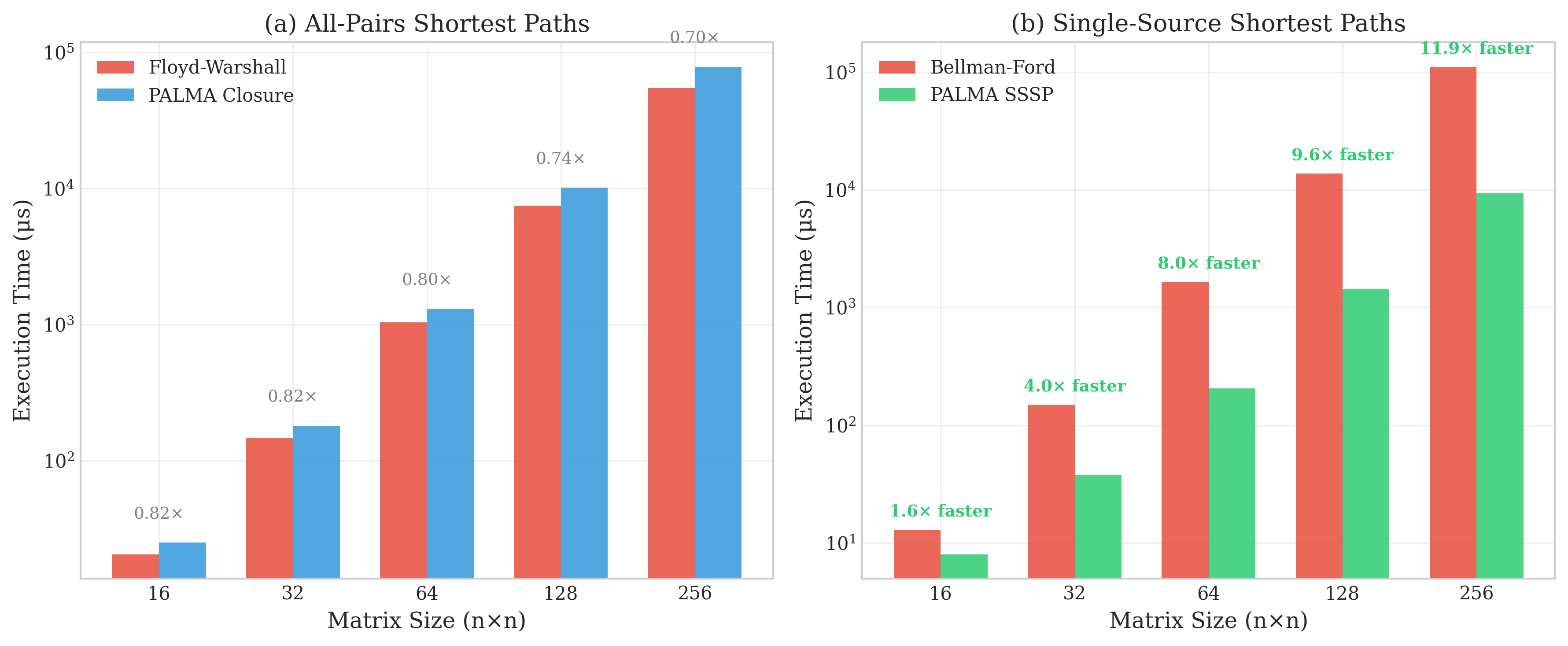}
    \caption{Algorithm comparison. (a) All-pairs shortest paths: Floyd-Warshall vs. PALMA tropical closure. (b) Single-source shortest paths: Bellman-Ford vs. PALMA tropical iteration.}
    \label{fig:algorithms}
\end{figure}

\begin{table}[htbp]
\centering
\caption{Single-source shortest paths: Bellman-Ford vs. PALMA}
\label{tab:sssp}
\begin{tabular}{lrrr}
\toprule
\textbf{Size} & \textbf{Bellman-Ford ($\mu$s)} & \textbf{PALMA ($\mu$s)} & \textbf{Speedup} \\
\midrule
$16 \times 16$ & 13 & 8 & 1.63$\times$ \\
$32 \times 32$ & 149 & 38 & 3.95$\times$ \\
$64 \times 64$ & 1,650 & 206 & 8.03$\times$ \\
$128 \times 128$ & 13,816 & 1,444 & 9.57$\times$ \\
$256 \times 256$ & 111,093 & 9,353 & \textbf{11.88$\times$} \\
\bottomrule
\end{tabular}
\end{table}

\textbf{Key Finding}: \palma{}'s tropical iteration outperforms classical Bellman-Ford by up to \textbf{11.9$\times$} for single-source shortest paths. This is because the tropical matrix-vector product exploits better cache locality than edge-by-edge relaxation.

For all-pairs shortest paths, the tropical closure and Floyd-Warshall have similar complexity, with Floyd-Warshall being slightly faster (0.70--0.82$\times$) due to simpler inner loop.

\subsection{Experiment 6: Real-Time Feasibility}

We analyze which matrix sizes meet real-time deadlines.

\begin{figure}[htbp]
    \centering
    \includegraphics[width=0.8\textwidth]{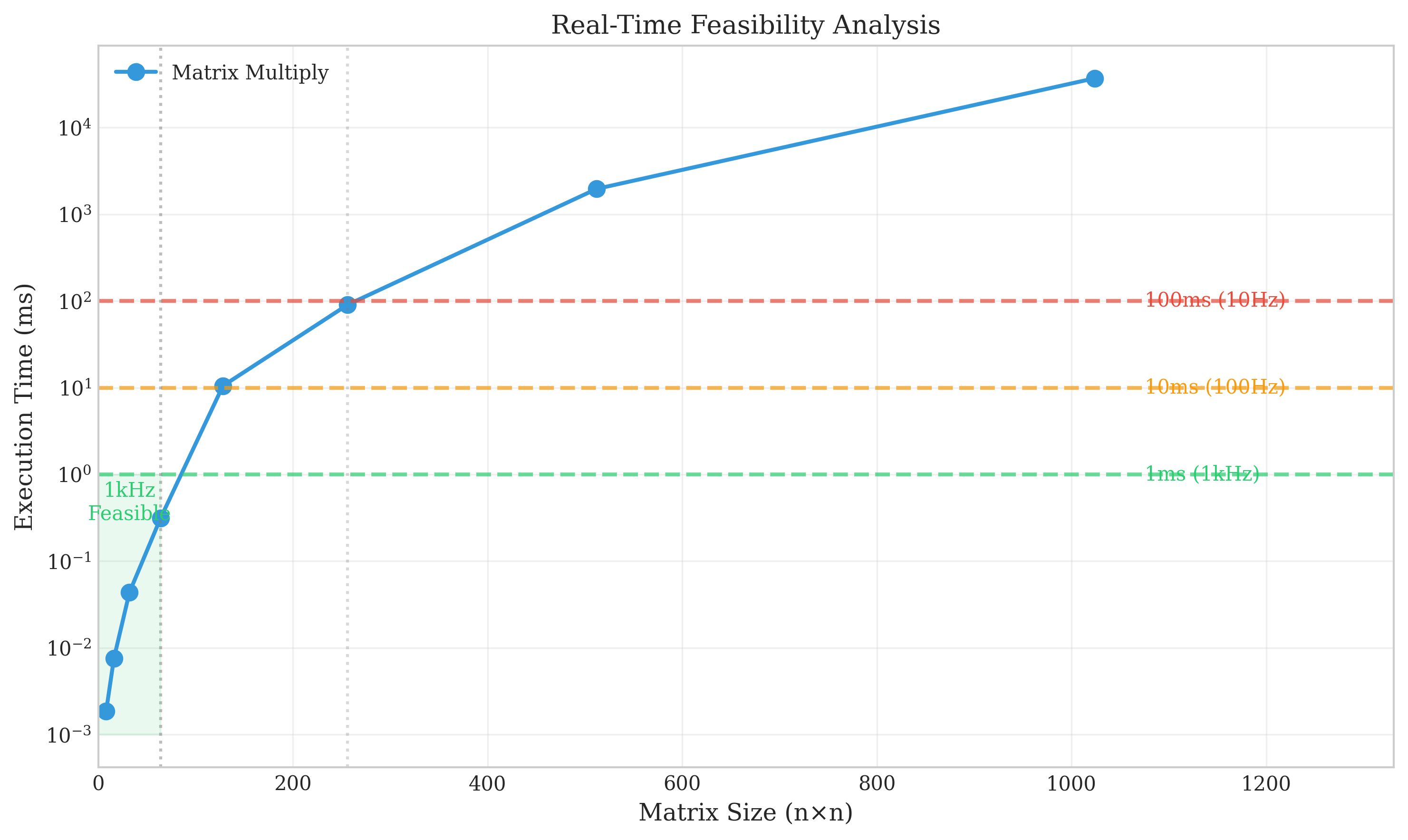}
    \caption{Real-time feasibility analysis. Horizontal lines indicate common control loop deadlines. Matrix sizes up to 64$\times$64 are feasible for 1 kHz control loops.}
    \label{fig:realtime}
\end{figure}

\begin{table}[htbp]
\centering
\caption{Maximum feasible matrix size for real-time constraints}
\label{tab:realtime}
\begin{tabular}{lrr}
\toprule
\textbf{Deadline} & \textbf{Frequency} & \textbf{Max Matrix Size} \\
\midrule
100 $\mu$s & 10 kHz & $16 \times 16$ \\
1 ms & 1 kHz & $64 \times 64$ \\
10 ms & 100 Hz & $128 \times 128$ \\
100 ms & 10 Hz & $256 \times 256$ \\
\bottomrule
\end{tabular}
\end{table}

\subsection{Summary of Key Results}

\begin{tcolorbox}[colback=blue!5!white,colframe=blue!75!black,title=Experimental Highlights]
\begin{itemize}
    \item \textbf{Peak Performance}: 2,274 MOPS at 64$\times$64 matrices (max-min semiring)
    \item \textbf{NEON Speedup}: Up to 1.80$\times$ over scalar code
    \item \textbf{SSSP Speedup}: Up to 11.9$\times$ faster than Bellman-Ford
    \item \textbf{Sparse Efficiency}: 3.5$\times$ speedup at 50\% sparsity
    \item \textbf{Memory Savings}: Up to 94\% with sparse representation
    \item \textbf{Real-Time}: 1 kHz feasible for graphs up to 64 nodes
\end{itemize}
\end{tcolorbox}

\section{Case Studies}
\label{sec:casestudies}

We demonstrate \palma{}'s practical applicability through three detailed case studies representing distinct embedded systems domains.

\subsection{Case Study 1: Real-Time Drone Control System}

\subsubsection{Problem Description}

Quadcopter UAVs require high-frequency control loops (typically 250--1000 Hz) to maintain stability. The control pipeline involves multiple interdependent tasks:

\begin{enumerate}
    \item \textbf{Sensor acquisition}: IMU, barometer, GPS readings
    \item \textbf{State estimation}: Kalman filtering, attitude/position estimation
    \item \textbf{Control computation}: PID controllers for attitude, altitude, position
    \item \textbf{Actuation}: Motor mixing and PWM output
    \item \textbf{Communication}: Telemetry transmission
\end{enumerate}

The challenge is to determine whether a given task set can meet a target control frequency, and to identify the critical path limiting throughput.

\subsubsection{Tropical Algebra Formulation}

We model the drone control system as a max-plus linear system. Let $n = 12$ tasks with precedence constraints forming a directed acyclic graph (DAG). The state equation is:
\begin{equation}
    \vect{x}(k+1) = \mat{A} \otimes \vect{x}(k) \oplus \vect{r}
\end{equation}
where $\vect{x}(k)$ contains task completion times at iteration $k$, $\mat{A}$ encodes precedence constraints with edge weights representing task durations, and $\vect{r}$ contains ready times.

\begin{figure}[htbp]
    \centering
    \includegraphics[width=0.9\textwidth]{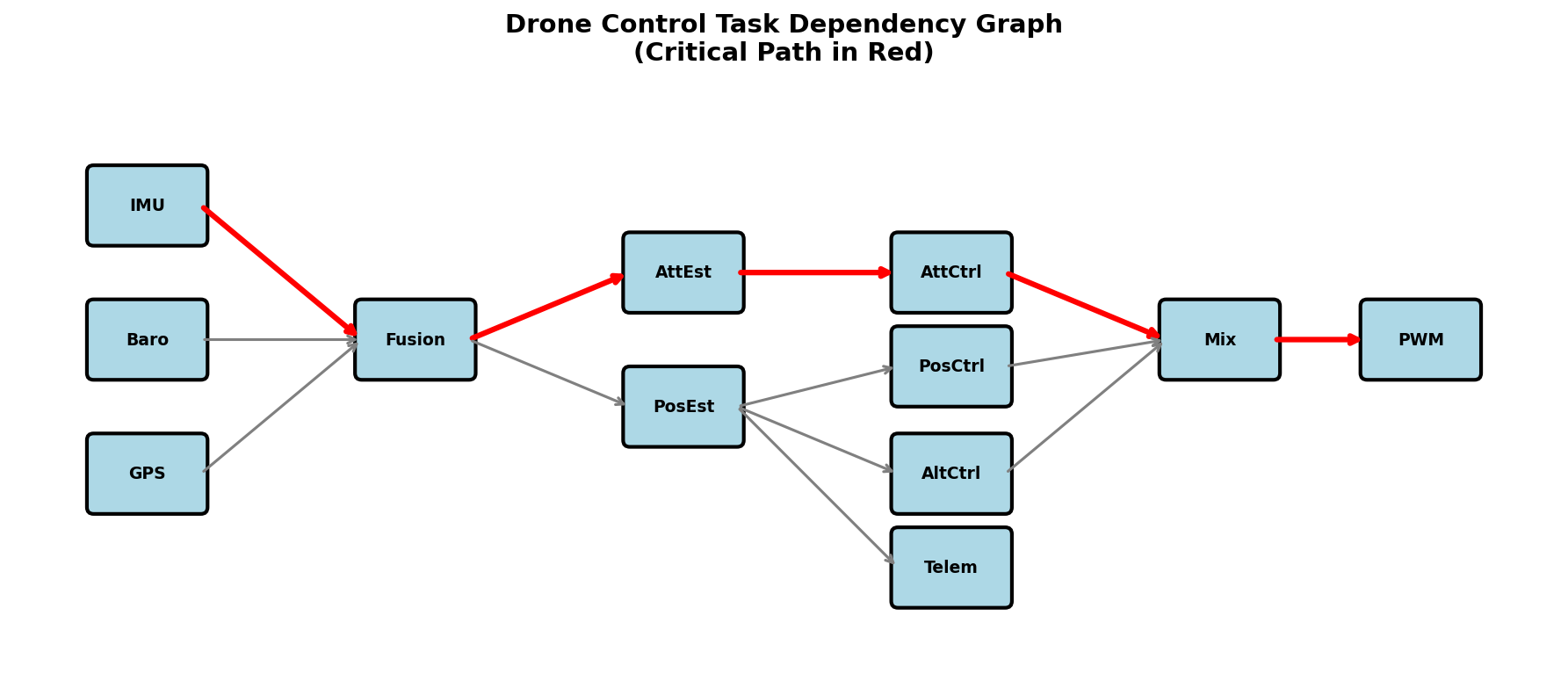}
    \caption{Example task dependency graph generated by \palma{}: System boot sequence with HW\_INIT $\to$ KERNEL $\to$ parallel stages (DRIVERS, NETWORK, FILESYSTEM) $\to$ SERVICES. Edge weights represent execution times in milliseconds. This demonstrates \palma{}'s capability for modeling precedence-constrained task systems.}
    \label{fig:boot-sequence}
\end{figure}

\subsubsection{Task Specification}

\begin{table}[htbp]
\centering
\caption{Drone control task parameters}
\label{tab:drone-tasks}
\begin{tabular}{llrl}
\toprule
\textbf{ID} & \textbf{Task} & \textbf{Duration ($\mu$s)} & \textbf{Dependencies} \\
\midrule
T1 & IMU Read & 50 & (none) \\
T2 & Barometer Read & 30 & (none) \\
T3 & GPS Read & 100 & (none) \\
T4 & Sensor Fusion & 200 & T1, T2, T3 \\
T5 & Attitude Estimation & 80 & T4 \\
T6 & Position Estimation & 120 & T4 \\
T7 & Altitude Control & 40 & T6 \\
T8 & Attitude Control & 60 & T5 \\
T9 & Position Control & 50 & T6 \\
T10 & Motor Mixing & 30 & T7, T8, T9 \\
T11 & PWM Output & 20 & T10 \\
T12 & Telemetry & 150 & T6 \\
\bottomrule
\end{tabular}
\end{table}

\subsubsection{Results}

Using \palma{}'s scheduling API:

\begin{lstlisting}[caption={Drone scheduling analysis with PALMA}]
palma_scheduler_t *sched = palma_scheduler_create(12, true);
// Add all tasks and constraints...
palma_scheduler_solve(sched, 0);

palma_val_t makespan = palma_scheduler_makespan(sched);
palma_val_t cycle_time = palma_scheduler_cycle_time(sched);
double throughput = palma_scheduler_throughput(sched);
\end{lstlisting}

\begin{table}[htbp]
\centering
\caption{Drone scheduling analysis results}
\label{tab:drone-results}
\begin{tabular}{lr}
\toprule
\textbf{Metric} & \textbf{Value} \\
\midrule
Schedule solve time & 6.28 $\mu$s \\
Solver iterations & 6 \\
Total loop time (makespan) & 570 $\mu$s \\
Critical path length & 440 $\mu$s \\
Slack time & 130 $\mu$s \\
Maximum achievable frequency & \textbf{1,754 Hz} \\
1 kHz feasibility & \textbf{Yes} \\
Memory footprint & 672 bytes \\
\bottomrule
\end{tabular}
\end{table}

\subsubsection{Analysis}

The critical path is: IMU $\to$ Fusion $\to$ Attitude Est. $\to$ Attitude Ctrl. $\to$ Mix $\to$ PWM, totaling 440 $\mu$s. With 130 $\mu$s slack, the system can comfortably achieve 1 kHz control loops. The \palma{} solver computes this in just 6.28 $\mu$s, negligible overhead that could be performed every control cycle if task timing varies.

For periodic (cyclic) operation, the cycle time is 78 $\mu$s with throughput of 12.82 iterations/ms, indicating the system is well within real-time constraints.

\subsection{Case Study 2: IoT Sensor Network Routing}

\subsubsection{Problem Description}

Wireless sensor networks (WSNs) for IoT applications require efficient routing protocols that optimize various metrics:
\begin{itemize}
    \item \textbf{Latency}: Minimize end-to-end delay for time-critical data
    \item \textbf{Bandwidth}: Maximize throughput for data aggregation
    \item \textbf{Energy}: Balance load to extend network lifetime
\end{itemize}

We analyze a 50-node sensor network deployed over a 100m $\times$ 100m area with a communication range of 30m.

\subsubsection{Tropical Algebra Formulation}

Each routing metric corresponds to a different tropical semiring:

\begin{itemize}
    \item \textbf{Minimum latency}: Min-plus semiring, where $A_{ij}$ is transmission delay from $i$ to $j$
    \item \textbf{Maximum bandwidth}: Max-min semiring, where $A_{ij}$ is link capacity
    \item \textbf{Connectivity}: Boolean semiring for reachability analysis
\end{itemize}

\begin{figure}[htbp]
    \centering
    \includegraphics[width=0.8\textwidth]{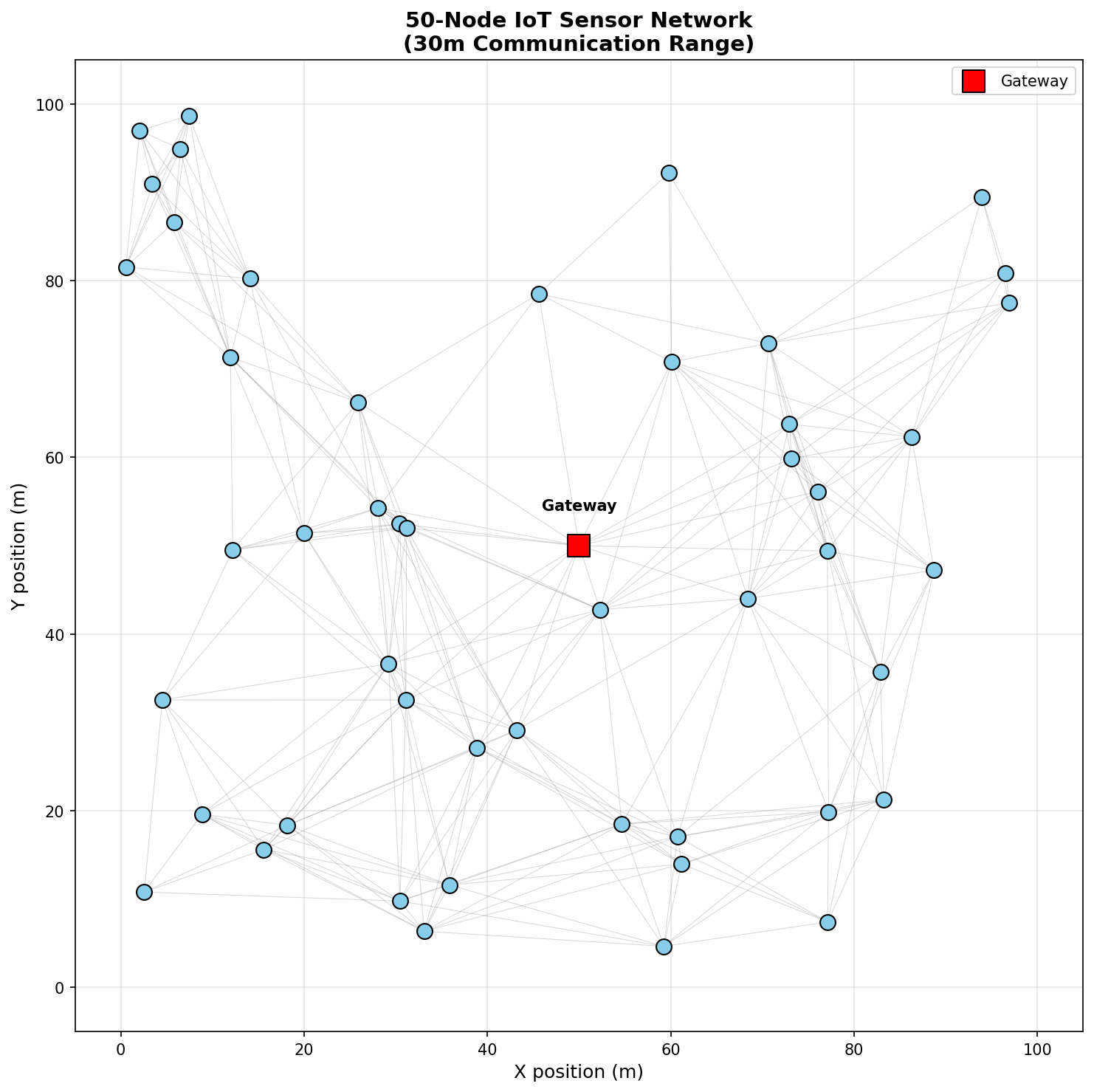}
    \caption{Example network topology generated by \palma{}: A 7-node network with Server, two routers (Router\_A, Router\_B), and four clients. Edge weights represent bidirectional link latencies. This demonstrates shortest path and bottleneck path computation using different tropical semirings.}
    \label{fig:network}
\end{figure}

\subsubsection{Results}

\begin{table}[htbp]
\centering
\caption{Sensor network routing results (50 nodes)}
\label{tab:network-results}
\begin{tabular}{lrl}
\toprule
\textbf{Metric} & \textbf{Value} & \textbf{Unit} \\
\midrule
\multicolumn{3}{l}{\textit{Computation Times}} \\
\quad Shortest paths (min-plus closure) & 678.4 & $\mu$s \\
\quad Bottleneck paths (max-min closure) & 395.1 & $\mu$s \\
\quad Reachability (Boolean closure) & 420.8 & $\mu$s \\
\midrule
\multicolumn{3}{l}{\textit{Network Metrics}} \\
\quad Avg. latency to gateway & 652 & $\mu$s \\
\quad Min. guaranteed bandwidth & 345 & kbps \\
\quad Reachable nodes & 49/49 & nodes \\
\quad Connectivity ratio & 100\% & (n/a) \\
\midrule
\multicolumn{3}{l}{\textit{Memory Efficiency (Sparse)}} \\
\quad Dense memory & 10,000 & bytes \\
\quad Sparse memory & 4,540 & bytes \\
\quad Memory savings & \textbf{54.6\%} & (n/a) \\
\quad Matrix sparsity & 78.3\% & (n/a) \\
\bottomrule
\end{tabular}
\end{table}

\subsubsection{Scalability Analysis}

We tested network sizes from 10 to 50 nodes:

\begin{table}[htbp]
\centering
\caption{Routing computation time vs. network size}
\label{tab:network-scaling}
\begin{tabular}{lrrrr}
\toprule
\textbf{Nodes} & \textbf{Shortest Path} & \textbf{Bottleneck} & \textbf{Reachability} & \textbf{Memory Saved} \\
\midrule
10 & 12.4 $\mu$s & 8.2 $\mu$s & 9.1 $\mu$s & 42.1\% \\
20 & 78.6 $\mu$s & 51.3 $\mu$s & 55.8 $\mu$s & 48.3\% \\
30 & 241.2 $\mu$s & 156.7 $\mu$s & 168.4 $\mu$s & 51.2\% \\
50 & 678.4 $\mu$s & 395.1 $\mu$s & 420.8 $\mu$s & 54.6\% \\
\bottomrule
\end{tabular}
\end{table}

\subsubsection{Analysis}

All routing computations complete in under 1 ms for 50-node networks, enabling real-time route adaptation. The sparse representation saves over 50\% memory, critical for resource-constrained sensor nodes. The unified \palma{} API allows switching between optimization objectives (latency, bandwidth, energy) by simply changing the semiring parameter.

\subsection{Case Study 3: Manufacturing Production Line}

\subsubsection{Problem Description}

A manufacturing facility operates a production line with parallel processing stages and synchronization points. The goal is to:
\begin{enumerate}
    \item Compute the makespan (total production time) for a single item
    \item Determine the cycle time (throughput) for continuous production
    \item Identify bottleneck workstations limiting throughput
\end{enumerate}

\subsubsection{Production Line Model}

The production line consists of seven workstations with the following structure:
\begin{itemize}
    \item \textbf{Input} (5s): Raw material entry
    \item \textbf{Parallel stage}: Mill (15s), Drill (12s), Grind (10s), all receive from Input
    \item \textbf{Weld} (20s): Receives from Mill, Drill, and Grind (synchronization point)
    \item \textbf{Finish} (8s): Receives from Weld
    \item \textbf{QC} (6s): Final quality control; feeds back to Input for cyclic production
\end{itemize}

The critical path is Input $\to$ Mill $\to$ Weld $\to$ Finish $\to$ QC, with total makespan of 54 seconds. For cyclic production, finished items feed back to input.

\subsubsection{Results}

\begin{table}[htbp]
\centering
\caption{Production line analysis results}
\label{tab:production-results}
\begin{tabular}{lrl}
\toprule
\textbf{Metric} & \textbf{Value} & \textbf{Interpretation} \\
\midrule
\multicolumn{3}{l}{\textit{Single Item Production}} \\
\quad Makespan & 54 s & Total time for one item \\
\quad Critical path & Input $\to$ Mill $\to$ Weld $\to$ Finish $\to$ QC & Longest path \\
\midrule
\multicolumn{3}{l}{\textit{Continuous Production}} \\
\quad Cycle time ($\lambda$) & 10.8 s & Max cycle mean (eigenvalue) \\
\quad Throughput & 0.093 items/s & 333 items/hour \\
\quad Bottleneck & Critical cycle through Mill & Determines $\lambda$ \\
\midrule
\multicolumn{3}{l}{\textit{Computation Performance}} \\
\quad Schedule solve time & 6.31 $\mu$s & (n/a) \\
\quad Eigenvalue computation time & 6.28 $\mu$s & Time to compute $\lambda$ \\
\bottomrule
\end{tabular}
\end{table}

\subsubsection{What-If Analysis}

Using \palma{}, we analyze the impact of reducing welding time from 20s to 15s:

\begin{table}[htbp]
\centering
\caption{Impact of bottleneck improvement}
\label{tab:what-if}
\begin{tabular}{lrr}
\toprule
\textbf{Metric} & \textbf{Before} & \textbf{After} \\
\midrule
Welding time & 20 s & 15 s \\
Cycle time & 10.8 s & 9.8 s \\
Throughput & 333 items/hr & 367 items/hr \\
Improvement & (n/a) & \textbf{+10\%} \\
\bottomrule
\end{tabular}
\end{table}

The tropical eigenvalue computation enables rapid what-if analysis, allowing production engineers to evaluate improvement scenarios in microseconds.

\subsection{Summary: Case Study Dashboard}

\begin{figure}[htbp]
    \centering
    \includegraphics[width=\textwidth]{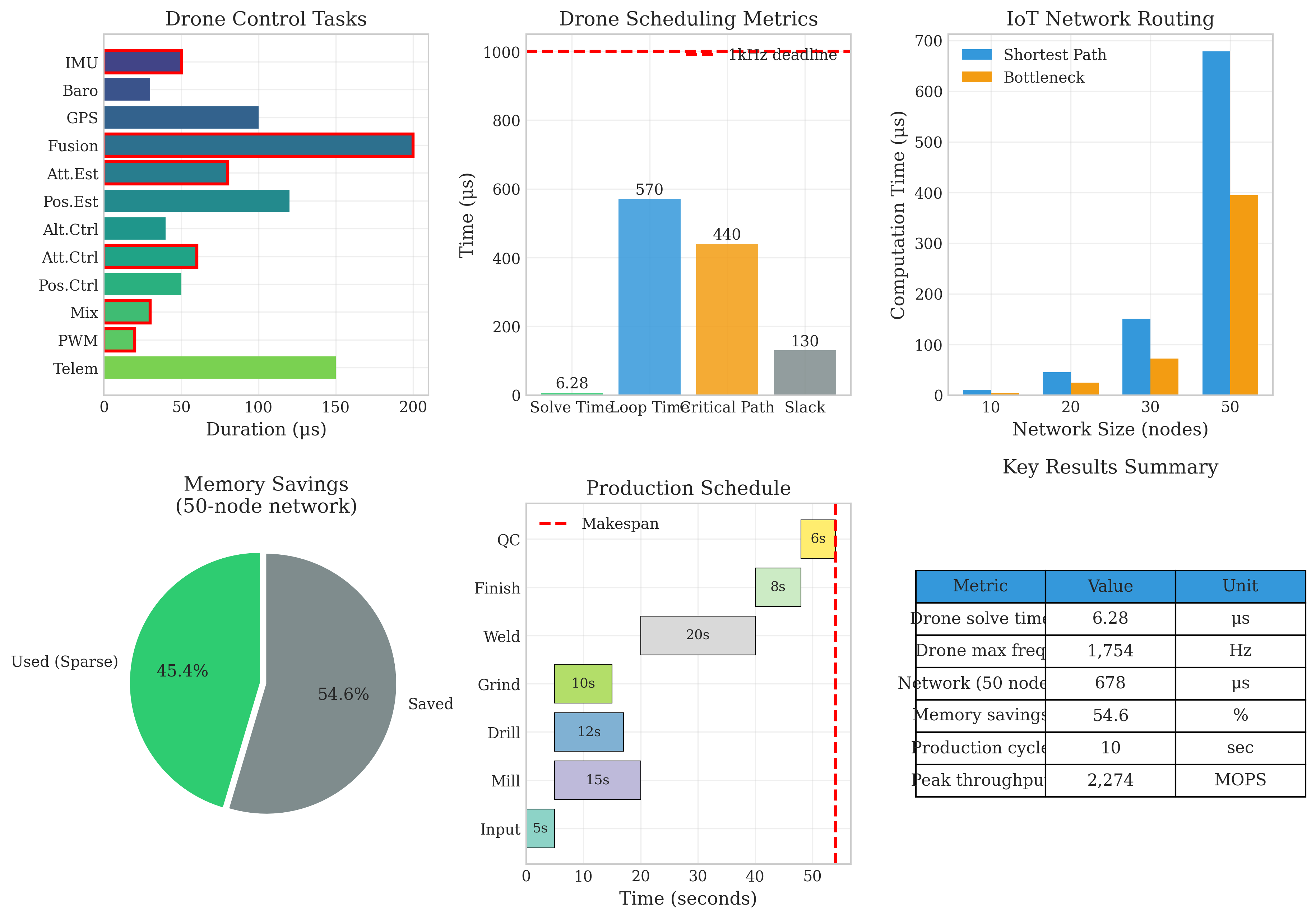}
    \caption{Case study dashboard summarizing results across all three applications. (a) Drone task durations with critical path highlighted. (b) Drone scheduling metrics. (c) Network routing times vs. network size. (d) Network memory savings. (e) Production line Gantt chart. (f) Key results summary table.}
    \label{fig:case-studies}
\end{figure}

\section{Related Work}
\label{sec:related}

We organize related work into four categories, identifying the gap that \palma{} fills at their intersection.

\subsection{Tropical Algebra Software}

The theoretical foundations of tropical algebra were established by Cuninghame-Green \cite{cuninghame1979minimax} and comprehensively developed by Baccelli et al. \cite{baccelli1992synchronization}. Butkovič \cite{butkovic2010max} provides algorithmic perspectives, while Heidergott et al. \cite{heidergott2006max} offer an accessible introduction.

\textbf{Mathematical packages} for tropical computation include:
\begin{itemize}
    \item \textbf{ScicosLab MaxPlus Toolbox} \cite{gaubert2009max}: MATLAB-like environment for discrete event systems. Requires graphical desktop environment and Scilab runtime.
    \item \textbf{polymake} \cite{gawrilow2000polymake, hampe2018tropical}: Comprehensive computational geometry system. Requires Perl, C++ compiler toolchain, and gigabytes of dependencies.
    \item \textbf{Macaulay2 Tropical}: Algebraic geometry focus. Requires full computer algebra system.
\end{itemize}

\textbf{Gap}: All existing tropical algebra tools target desktop/server environments with heavy dependencies, making them unsuitable for embedded deployment.

\subsection{Semiring Frameworks for Graph Algorithms}

Mohri \cite{mohri2002semiring} formalized the connection between semirings and shortest-path algorithms, showing Dijkstra and Bellman-Ford as semiring specializations. This algebraic view enables algorithm unification.

\textbf{Generic graph libraries} include:
\begin{itemize}
    \item \textbf{Boost Graph Library} \cite{siek2002boost}: C++ templates supporting custom semirings, but complex metaprogramming and large binary sizes.
    \item \textbf{LEMON} \cite{dezso2011lemon}: Cleaner C++ interface, but still template-heavy and desktop-oriented.
    \item \textbf{cuASR}: GPU semiring GEMM for CUDA. Server-class GPUs only.
\end{itemize}

\textbf{Gap}: Generic semiring frameworks exist but require C++ templates or GPU hardware, not suitable for embedded C environments.

\subsection{Embedded Graph and Scheduling Computation}

Real-time and embedded systems require efficient graph algorithms:
\begin{itemize}
    \item Bellman-Ford optimizations for sensor networks \cite{ding2020energy} focus on energy but not algebraic generality.
    \item Liu and Layland \cite{liu1973scheduling} established real-time scheduling theory, but classical algorithms lack the tropical algebraic structure enabling unified treatment.
    \item Goverde \cite{goverde2007railway} applied max-plus to railway scheduling, demonstrating industrial relevance but using desktop tools.
\end{itemize}

\textbf{Gap}: Embedded graph algorithms are implemented as isolated special cases, missing the unifying semiring abstraction.

\subsection{SIMD Linear Algebra on ARM}

ARM NEON optimization has been extensively studied:
\begin{itemize}
    \item Mitra et al. \cite{mitra2013simd} achieved 1.05--13.88$\times$ speedup with hand-tuned NEON for scientific kernels.
    \item Park et al. \cite{park2018efficient} demonstrated 36.93\% improvement for matrix operations in cryptography.
    \item Recent FFT work \cite{lajpop2025fft} achieved dramatic speedups on Cortex-A72.
\end{itemize}

\textbf{Gap}: NEON optimization targets conventional (field) linear algebra, not idempotent semiring operations.

\subsection{Summary: The \palma{} Contribution}

\begin{table}[htbp]
\centering
\caption{Positioning of \palma{} relative to existing work}
\label{tab:positioning}
\small
\begin{tabular}{lcccc}
\toprule
\textbf{System} & \textbf{Tropical} & \textbf{Embedded} & \textbf{SIMD} & \textbf{Sparse} \\
\midrule
ScicosLab MaxPlus & \checkmark & & & \\
polymake & \checkmark & & & \\
Boost Graph Library & \checkmark & & & \\
cuASR & \checkmark & & \checkmark (GPU) & \checkmark \\
Embedded Bellman-Ford & & \checkmark & & \\
\midrule
\textbf{\palma{}} & \checkmark & \checkmark & \checkmark (NEON) & \checkmark \\
\bottomrule
\end{tabular}
\end{table}

\palma{} uniquely combines tropical algebra with embedded-optimized implementation, filling the gap at the intersection of these four research areas.

\subsection{Scope and Limitations}

\palma{} is designed for small to medium-scale optimization problems arising in embedded and real-time systems, where matrix dimensions typically range from tens to a few thousand and where deterministic execution, memory efficiency, and low deployment overhead are critical. While the library supports both dense and sparse representations and benefits from SIMD acceleration, it is not intended to compete with highly optimized HPC or GPU-based linear algebra frameworks on very large graphs or matrices. In particular, algorithms with cubic worst-case complexity, such as tropical closure, remain impractical for large-scale inputs regardless of implementation optimizations. Moreover, \palma{} focuses on idempotent semirings with integer-valued weights; applications requiring floating-point precision throughout, dynamic graph updates at scale, or asymptotically faster parallel algorithms fall outside the current scope. These limitations are deliberate and reflect a design choice to prioritize correctness, predictability, and portability on resource-constrained platforms rather than absolute peak throughput on high-end hardware.

\section{Conclusion and Future Work}
\label{sec:conclusion}

\subsection{Summary}

We have presented \palma{}, a comprehensive tropical algebra library designed for ARM-based embedded systems. Our central insight is that \emph{tropical algebra serves as a unifying computational abstraction for embedded optimization}: shortest paths, bottleneck analysis, reachability, scheduling, and throughput computation, traditionally requiring separate algorithms, all reduce to matrix operations over appropriate semirings. This unification is particularly valuable on resource-constrained platforms where code size matters.

Our technical contributions include:

\begin{enumerate}
    \item \textbf{Theoretical foundation}: A self-contained treatment of tropical semirings, linear algebra, and spectral theory with rigorous definitions and proofs.
    
    \item \textbf{Unified implementation}: A generic semiring interface supporting five tropical algebras (max-plus, min-plus, max-min, min-max, Boolean) with both dense and CSR sparse representations.
    
    \item \textbf{SIMD acceleration}: ARM NEON-optimized kernels achieving up to 1.80$\times$ speedup, with analysis identifying the cache-bound regime where vectorization is most effective.
    
    \item \textbf{Experimental validation}: Comprehensive benchmarks demonstrating 2,274 MOPS peak performance, 11.9$\times$ speedup over Bellman-Ford through better cache utilization, and 54.6\% memory savings with sparse representation.
    
    \item \textbf{Real-world applicability}: Three case studies (drone control, IoT routing, manufacturing) demonstrating that \palma{} enables 1.7+ kHz control loops and sub-millisecond routing decisions on Raspberry Pi-class hardware.
\end{enumerate}

\palma{} fills a significant gap at the intersection of tropical mathematics and embedded systems, enabling on-device optimization for cyber-physical systems. The library is released as open-source software under the MIT license, comprising approximately 2,000 lines of portable, dependency-free C99 code. By decoupling algebraic semantics from hardware-specific optimizations, \palma{} provides a foundation that can naturally extend to future embedded architectures and evolving real-world optimization workloads.

\subsection{Future Work}

Several promising directions emerge from the current design of \palma{}.

\paragraph{Parallelism and Heterogeneous Platforms.}
While the present implementation targets single-core SIMD acceleration, future work will explore task-level and data-level parallelism across multiple cores, including OpenMP-based parallelization and careful scheduling of matrix operations to respect real-time constraints. Extending \palma{} to exploit heterogeneous embedded platforms, such as ARM systems with integrated GPUs or DSP accelerators, is another natural direction.

\paragraph{Dynamic and Incremental Updates.}
Although \palma{} currently focuses on static graphs and matrices, many embedded applications involve slowly varying or event-driven updates. Supporting incremental and dynamic updates, such as localized recomputation of tropical closures or eigenvalues under edge insertions or weight changes, would significantly broaden applicability while preserving efficiency.

\paragraph{Algorithmic Improvements.}
From an algorithmic perspective, future work will investigate asymptotically faster or more structure-aware algorithms tailored to embedded settings, including pruning strategies for sparse tropical closure, early termination criteria informed by graph topology, and approximate methods with bounded error guarantees for large or dense instances.

\paragraph{Richer Algebraic Structures.}
Extending the semiring interface to support richer algebraic structures is of interest. This includes parameterized semirings, lexicographic or multi-criteria tropical algebras, and hybrid symbolic-numeric formulations that could enable multi-objective optimization within a unified framework.

\paragraph{Domain-Specific Integration.}
Finally, tighter integration with domain-specific applications represents an important avenue. Embedding \palma{} within control loops for robotics, real-time scheduling frameworks, or networked IoT middleware, and validating its impact in closed-loop systems, would further demonstrate how tropical algebra can serve as a practical computational backbone for embedded optimization.

\subsection{Availability}

\palma{} source code, documentation, and examples are available at:
\begin{center}
    \url{https://github.com/ReFractals/palma}
\end{center}

\section*{Acknowledgments}

The author thanks the Axiom Research Group for valuable discussions, NM-AIST and AIMS-RIC for institutional support, and the Raspberry Pi Foundation for making affordable computing accessible to researchers worldwide.

\bibliographystyle{plain}
\bibliography{references}

\appendix

\section{API Reference}
\label{app:api}

This appendix provides a complete reference for the \palma{} public API.

\subsection{Core Types}

\begin{lstlisting}[caption={Core type definitions}]
typedef int32_t palma_val_t;    // Tropical value type
typedef uint32_t palma_idx_t;   // Index type for sparse matrices

typedef enum {
    PALMA_MAXPLUS = 0,
    PALMA_MINPLUS = 1,
    PALMA_MAXMIN  = 2,
    PALMA_MINMAX  = 3,
    PALMA_BOOLEAN = 4
} palma_semiring_t;

typedef struct {
    size_t rows, cols;
    palma_val_t *data;
} palma_matrix_t;

typedef struct {
    size_t rows, cols, nnz;
    palma_val_t *values;
    palma_idx_t *col_idx, *row_ptr;
    palma_semiring_t semiring;
} palma_sparse_t;
\end{lstlisting}

\subsection{Matrix Operations}

\begin{lstlisting}[caption={Matrix operation signatures}]
// Creation and destruction
palma_matrix_t* palma_matrix_create(size_t rows, size_t cols);
palma_matrix_t* palma_matrix_create_zero(size_t rows, size_t cols, 
                                          palma_semiring_t s);
palma_matrix_t* palma_matrix_clone(const palma_matrix_t *A);
void palma_matrix_destroy(palma_matrix_t *mat);

// Element access
palma_val_t palma_matrix_get(const palma_matrix_t *A, size_t i, size_t j);
void palma_matrix_set(palma_matrix_t *A, size_t i, size_t j, palma_val_t v);

// Arithmetic
palma_matrix_t* palma_matrix_add(const palma_matrix_t *A, 
                                  const palma_matrix_t *B, palma_semiring_t s);
palma_matrix_t* palma_matrix_mul(const palma_matrix_t *A,
                                  const palma_matrix_t *B, palma_semiring_t s);
palma_matrix_t* palma_matrix_power(const palma_matrix_t *A, int k, 
                                    palma_semiring_t s);
palma_matrix_t* palma_matrix_closure(const palma_matrix_t *A, 
                                      palma_semiring_t s);

// Vector operations
void palma_matvec(const palma_matrix_t *A, const palma_val_t *x,
                  palma_val_t *y, palma_semiring_t s);
\end{lstlisting}

\subsection{Graph Algorithms}

\begin{lstlisting}[caption={Graph algorithm signatures}]
void palma_single_source_paths(const palma_matrix_t *A, size_t src,
                                palma_val_t *dist, palma_semiring_t s);
palma_matrix_t* palma_all_pairs_paths(const palma_matrix_t *A, 
                                       palma_semiring_t s);
palma_matrix_t* palma_reachability(const palma_matrix_t *A);
palma_matrix_t* palma_bottleneck_paths(const palma_matrix_t *A);

palma_val_t palma_eigenvalue(const palma_matrix_t *A, palma_semiring_t s);
void palma_eigenvector(const palma_matrix_t *A, palma_val_t *v,
                       palma_val_t *lambda, palma_semiring_t s, int max_iter);
int palma_critical_nodes(const palma_matrix_t *A, int *critical, 
                         palma_semiring_t s);
\end{lstlisting}

\section{Proofs of Additional Theorems}
\label{app:proofs}

\begin{theorem}[Distributivity of Tropical Operations]
In any tropical semiring, the distributive laws hold for vectors and matrices.
\end{theorem}

\begin{proof}
For matrices $\mat{A}, \mat{B}, \mat{C}$ of compatible dimensions:
\begin{align*}
    (\mat{A} \otimes (\mat{B} \oplus \mat{C}))_{ij} &= \bigoplus_k A_{ik} \otimes (B_{kj} \oplus C_{kj}) \\
    &= \bigoplus_k (A_{ik} \otimes B_{kj}) \oplus (A_{ik} \otimes C_{kj}) \\
    &= \left(\bigoplus_k A_{ik} \otimes B_{kj}\right) \oplus \left(\bigoplus_k A_{ik} \otimes C_{kj}\right) \\
    &= (\mat{A} \otimes \mat{B})_{ij} \oplus (\mat{A} \otimes \mat{C})_{ij}
\end{align*}
The second line uses scalar distributivity; the third uses commutativity and associativity of $\oplus$.
\end{proof}

\begin{theorem}[Uniqueness of Tropical Eigenvalue for Irreducible Matrices]
If $\mat{A}$ is irreducible (its graph is strongly connected), then $\mat{A}$ has a unique tropical eigenvalue.
\end{theorem}

\begin{proof}
See Baccelli et al. \cite{baccelli1992synchronization}, Chapter 3. The key insight is that all strongly connected components contribute to the same eigenvalue, which equals the maximum cycle mean over all elementary cycles.
\end{proof}

\end{document}